\documentclass[final,12pt]{article}
\usepackage{amsfonts,color,morefloats,pslatex,a4wide}
\usepackage{amssymb,amsthm,amsmath,latexsym,pslatex}

\newtheorem{theorem}{Theorem}
\newtheorem{lemma}[theorem]{Lemma}

\newtheorem{corollary}[theorem]{Corollary}

\newtheorem{open}[theorem]{Open Problem}

\newtheorem{example}[theorem]{Example}
\newtheorem{question}[theorem]{Question}

\usepackage[para]{threeparttable}

\newcommand{\ord}{{\mathrm{ord}}}

\newcommand{\tr}{{\mathrm{Tr}}}

\newcommand{\gf}{{\mathrm{GF}}}
\newcommand{\PG}{{\mathrm{PG}}}

\newcommand{\GA}{{\mathrm{GA}}}

\newcommand{\PAut}{{\mathrm{PAut}}}
\newcommand{\MAut}{{\mathrm{MAut}}}
\newcommand{\GAut}{{\mathrm{Aut}}}

\newcommand{\PRM}{{\mathrm{PRM}}}

\newcommand{\AG}{{\mathrm{AG}}}

\newcommand{\wt}{{\mathtt{wt}}}

\newcommand{\Z}{\mathbb{{Z}}}

\newcommand{\m}{\mathbb{M}}

\newcommand{\C}{{\mathcal{C}}}

\newcommand{\cR}{{\mathcal{R}}}

\newcommand{\bc}{{\mathbf{c}}}

\newcommand{\bzero}{{\mathbf{0}}}

\newcommand{\0}{\textbf{0}}

\makeatletter

\newcommand{\Rmnum}[1]{\expandafter\@slowromancap\romannumeral #1@}
\makeatother
\newcommand\myatop[2]{\genfrac{}{}{0pt}{}{#1}{#2}}

\begin{document} 

\title{Two Classes of Constacyclic Codes with Variable Parameters\thanks{
Z. Sun's research was supported by The National Natural Science Foundation of China under Grant Number  62002093. 
C. Ding's research was supported by The Hong Kong Research Grants Council, Proj. No. $16301522$,
X. Wang's research was supported by The National Natural Science Foundation of China under Grant Number  12001175. 
}\footnote{A small part of this paper was presented at the International Workshop on the Arithmetic of Finite Fields (WAIFI 2022), Chengdu, China. August 29-September 2, 2022}}

\author{Zhonghua Sun\thanks{School of Mathematics, Hefei University of Technology, Hefei, 230601, Anhui, China. Email:  sunzhonghuas@163.com}  
\and Cunsheng Ding\thanks{Department of Computer Science
                           and Engineering, The Hong Kong University of Science and Technology,
Clear Water Bay, Kowloon, Hong Kong, China. Email: cding@ust.hk} 
\and Xiaoqiang Wang\thanks{Hubei Key Laboratory of Applied Mathematics, Faculty of Mathematics and Statistics, Hubei University, Wuhan 430062, China.  
Email: waxiqq@163.com}  
}

\maketitle

\begin{abstract} 
Constacyclic codes over finite fields are a family of linear codes and contain cyclic codes as a subclass. Constacyclic codes are related to many areas of mathematics and outperform cyclic codes in several aspects. Hence, constacyclic codes are of theoretical 
importance. On the other hand, constacyclic codes are important in practice, as they have rich algebraic structures and may 
have efficient decoding algorithms.  
In this paper, two classes of constacyclic codes are constructed using a general construction of constacyclic codes with cyclic codes.  The first class of constacyclic codes is motivated by the 
punctured Dilix cyclic codes and the second class is motivated by the punctured generalised Reed-Muller codes. The two classes of 
constacyclic codes contain optimal linear codes. The parameters of the two classes of constacyclic codes are analysed and some 
open problems are presented in this paper.    
\end{abstract}

\section{Introduction and motivations} 

\subsection{Constacyclic codes and cyclic codes}

Throughout this paper, let $q$ be a prime power, $\gf(q)$ be the finite field with $q$ elements, and let $\gf(q)^*$ denote the multiplicative group of $\gf(q)$. By an $[n, k, d]$ code $\C$ over $\gf(q)$ we mean a $k$-dimensional linear subspace of $\gf(q)^n$ with minimum distance $d$. Let $\lambda \in \gf(q)^*$. We say that a linear code $\C$ of length $n$ is $\lambda$-{\it constacyclic} if $(c_0,c_1,\ldots,c_{n-1})\in \C$ implies $(\lambda c_{n-1}, c_0,c_1,\ldots,c_{n-2})\in \C$. Let $\Phi$ be the mapping from $\gf(q)^n$ to the quotient ring $\gf(q)[x]/\langle x^n-\lambda \rangle$ defined by 
\[\Phi((c_0,c_1, \ldots, c_{n-1})) = \sum_{i=0}^{n-1} c_i x^i.\]  
It is well known that every ideal of the ring $\gf(q)[x]/\langle x^n-\lambda \rangle$ is {\it principal} and a linear code $\C \subset \gf(q)^n$ is 
$\lambda$-constacyclic if and only if $\Phi(\C)$ is an ideal of $\gf(q)[x]/ \langle x^n-\lambda \rangle$. Consequently, we will 
identify $\C$ with $\Phi(\C)$ for any $\lambda$-constacyclic code $\C$. 
Let $\C=\langle g(x) \rangle$ be a 
$\lambda$-constacyclic code over $\gf(q)$, where $g(x)$ is monic and has the smallest degree. Then $g(x)$ 
is called the {\it generator polynomial} and $h(x)=(x^n-\lambda)/g(x)$ is referred to as the {\it check polynomial} of $\C$. The dual code of $\C$, denoted by $\C^{\bot}$, is defined by
$$\C^{\bot}=\{\mathbf{b}\in \gf(q)^n: \mathbf{b}\mathbf{c}^T=0,\ \forall \ \bc\in \C \},$$
where $\mathbf{b}\bc^T$ denotes the standard inner product of the two vectors $\mathbf{b}$ and $\bc$. The dual code $\C^\perp$ of $\C$ is the $\lambda^{-1}$-constacyclic code generated by the reciprocal polynomial of the check polynomial $h(x)$ of $\C$. By definition, $1$-constacyclic codes are the classical cyclic codes. Hence, cyclic codes form a subclass of constacyclic codes. In other words, constacyclic codes are a generalisation of the classical cyclic codes. For more information on constacyclic codes, the reader is referred to \cite{Black66,CDFL15,CFLL12,DP92,DDR11,DY10,KS90,LLLM17,LQL2017,LQ2018,MC21,PD91,SR2018,WSZ19,Wolfmann2008,SZW20} 
and the references therein.

\subsection{Motivations and objectives}

By definition, cyclic codes are a proper subclass of constacyclic codes and constacyclic codes are a proper subclass of linear 
codes. Clearly, cyclic codes have a better algebraic structure than $\lambda$-constacyclic codes with $\lambda \neq 1$ and 
constacyclic codes have a better algebraic structure than other linear codes. A better algebraic structure may mean a better 
decoding algorithm. Then the following two questions are interesting and good motivations for studying constacyclic codes. 

\begin{question} 
Is a given linear code over $\gf(q)$ monomially-equivalent to a cyclic code over $\gf(q)$? 
\end{question} 

\begin{question} 
Is a given linear code over $\gf(q)$ monomially-equivalent to a $\lambda$-constacyclic code over $\gf(q)$ with $\lambda \neq 1$? 
\end{question} 

For example, the Hamming code of length $(q^m-1)/(q-1)$ over $\gf(q)$ is monomially-equivalent to a cyclic code over $\gf(q)$ 
when $\gcd(m, q-1)=1$, and is always monomially-equivalent to a contacyclic code over $\gf(q)$. This shows that the Hamming 
code is attractive. Notice that the two questions above are open for most linear codes. 

Recall that cyclic codes have a better algebraic structure. Then one would ask why we would study constacyclic codes. 
Below is a list of motivations for studying $\lambda$-constacyclic codes with $\lambda \neq 1$: 
\begin{itemize}
\item There does not exist a cyclic code over $\gf(q)$ with parameters $[n, k, d]$ for certain $q$, $n$, $k$ and $d$; 
          but there is a $\lambda$-constacyclic codes over $\gf(q)$ with parameters $[n, k, d]$ and $\lambda \neq 1$. 
\item The best $[n, k]$ constacyclic code over $\gf(q)$ has a much better error-correcting capability than 
          the best $[n, k]$ cyclic code over $\gf(q)$ for certain $q$, $n$ and $k$. 
\item Constacyclic codes can do many things that cyclic codes cannot do. For example, the Hamming code of length 
$(q^m-1)/(q-1)$ can always be 
constructed by a constacyclic code, but cannot be constructed by a cyclic code when $\gcd(q-1, m) \neq 1$.                      
\end{itemize}

The original binary Reed-Muller codes were introduced by Reed and Muller in 1954 \cite{Reed54,Muller54}. They are called 
geometric codes, as all the minimum weight codewords of the $r$-th order Reed-Muller code $\cR_2(r, m)$ are the incidence vectors of 
all the $(m-r)$-flats in the affine geometry $\AG(m, \gf(2))$ and they generate $\cR_2(r, m)$ \cite{AK98}. The automorphism group of $\cR_2(r, m)$ is known to be the general affine group $\GA_m(\gf(2))$ which is triply transitive on $\gf(2)^m$. Hence, the binary 
Reed-Muller codes support $3$-designs. It was later discovered that the binary Reed-Muller codes become cyclic codes if they are punctured in a special coordinate position.  These properties show that the original Reed-Muller codes are very interesting in theory. 
Binary Reed-Muller codes are also interesting in practice as they have efficient decoding algorithms \cite{Reed54}. The binary 
Reed-Muller codes and their punctured codes were later generalised to codes over $\gf(q)$ for general $q$. In 2018, the binary Reed-Muller 
codes were generalised into another type of linear codes \cite{DLX18}, which were called \emph{Dilix codes} for the purpose of distinguishing the two types of generalisations \cite[Chapter 6]{dingtang2022}. The Dilix codes have also interesting properties and are extended cyclic codes by definition. In other words, if the Dilix codes are punctured in the last coordinate, the punctured Dilix codes are cyclic. Motivated by the punctured generalized Reed-Muller codes and punctured Dilix codes, we will construct and analyse two classes of constacyclic codes which are counterparts 
of  the punctured generalized Reed-Muller codes and the punctured Dilix codes.  
%As will be shown later, the two classes of codes contain optimal codes.  

 \subsection{The organisation of this paper}

The rest of this paper is organized as follows. Section \ref{sec-prel} recalls some basic results about linear codes and 
constacyclic codes, which will be needed later. Section \ref{sec-3} introduces a general construction of constacyclic codes of length $\frac{q^m-1}{r}$ with cyclic codes of length $q^m-1$. Section \ref{sec-1stclass} introduces the first class of constacyclic codes and analyses the parameters of these codes. Section \ref{sec-2ndclass} introduces the second class of constacyclic codes and analyses the parameters of these codes. Section \ref{sec4} concludes this paper and makes concluding remarks.

\section{Preliminaries}\label{sec-prel} 

\subsection{Some basic notation}

Throughout this paper, we fix the following notation, unless it is stated otherwise: 
\begin{itemize}
\item $q$ is a prime power. 
\item $m \geq 2$ is an integer. 
\item $r \geq 2$ is a positive divisor of $q-1$. 
\item $N=q^m-1$.  
\end{itemize}

For a linear code $\C$, we use $\dim(\C)$ and $d(\C)$ to denote its dimension and minimum Hamming distance, respectively.  
For a linear code $\C \subset \gf(q)^n$, let $A_i$ denote the number of codewords with Hamming weight $i$ in $\mathcal{C}$. The {\it weight enumerator} of $\C$ is defined as $1+A_1z+\cdots+A_nz^n$. The sequence $(1,A_1,\ldots,A_n)$ is called the {\it weight distribution} of $\C$. If the number of nonzero $A_i$ in the sequence $(A_1,A_2,\ldots, A_n)$ equals $t$, then $\C$ is called a $t$-weight code.  

\subsection{The Hamming weight and $q$-weight of nonnegative integers} 

Let $N=q^m-1$. 
For each $0 \leq i \leq N$, let the $q$-adic expansion of $i$ be $i=\sum_{j=0}^{m-1} i_j q^j$, where $0 \leq i_j \leq q-1$.  
The \emph{Hamming weight} of $i$, denoted by $\wt(i)$, is defined to be the Hamming weight of the vector 
$(i_0, i_1, \ldots, i_{m-1})$. The \emph{ $q$-weight} of $i$, denoted by $\wt_q(i)$, is defined to be $\sum_{j=0}^{m-1} i_j$. 
The two kinds of weights will be used later. 

\subsection{Some bounds of linear codes}

We now recall two bounds on linear codes, which will be needed later.

\begin{lemma}\label{lem-SPB1}
{\rm (Sphere Packing Bound \cite{HP2003})} Let $\C$ be an $[n, k, d]$ code over $\gf(q)$, then 
$$\sum_{i=0}^{\lfloor \frac{d-1}2\rfloor} \binom{n}{i} (q-1)^i\leq q^{n-k},$$
where $\lfloor \cdot \rfloor$ is the floor function.
\end{lemma}

The following lemma is the sphere packing bound for linear codes with an even minimum distance.

\begin{lemma}{\rm \cite{FWF2017}}\label{lem-SPB2}
Let $\C$ be an $[n,k,d]$ code over $\gf(q)$, where $d$ is an even integer.	Then
$$\sum_{i=0}^{\frac{d-2}2} \binom{n-1}{i} (q-1)^i\leq q^{n-1-k}.$$
\end{lemma}

\subsection{Cyclotomic cosets} 

Let $n$ be a positive integer with $\gcd(q, n)=1$, $r$ be a positive divisor of $q-1$, and let $\lambda$ be an element of $\gf(q)$ with order $r$. To deal with $\lambda$-constacyclic codes of length $n$ over $\gf(q)$, we have to study the factorization of $x^n-\lambda$ over $\gf(q)$. To this end, we need to introduce $q$-cyclotomic cosets modulo $rn$. 

Let $\Z_{rn}=\{0,1,2,\cdots,rn-1\}$ be the ring of integers modulo $rn$. For any $i \in \Z_{rn}$, the \emph{$q$-cyclotomic coset of $i$ modulo $rn$} is defined by 
\[C^{(q,rn)}_i=\{i, iq, iq^2, \cdots, iq^{\ell_i-1}\} \bmod {rn} \subseteq \Z_{rn}, \]
where $\ell_i$ is the smallest positive integer such that $i \equiv i q^{\ell_i} \pmod{rn}$, and is the \textit{size} of the $q$-cyclotomic coset $C^{(q,rn)}_i$. The smallest integer in $C^{(q, rn)}_i$ is called the \textit{coset leader} of $C^{(q, rn)}_i$. Let $\Gamma_{(q, rn)}$ be the set of all the coset leaders. We have then $C^{(q, rn)}_i \cap C^{(q, rn)}_j = \emptyset$ for any two distinct elements $i$ and $j$ in  $\Gamma_{(q, rn)}$, and  
 $\bigcup_{i \in  \Gamma_{(q, rn)} } C_i^{(q,rn)}=\Z_{rn}$.

Let $m=\ord_{r n}(q)$. It is easily seen that there is a primitive element $\alpha$ of $\gf(q^m)$ such that $\beta=\alpha^{(q^m-1)/rn}$ and $\beta^n=\lambda$. Then $\beta$ is a primitive $rn$-th root of unity in $\gf(q^m)$. 
 The \textit{minimal polynomial} $\m_{\beta^i}(x)$ of $\beta^i$ over $\gf(q)$ is the monic polynomial of the smallest degree over $\gf(q)$ with $\beta^i$ as a zero. We have $
\m_{\beta^i}(x)=\prod_{j \in C_i^{(q,rn)}} (x-\beta^j) \in \gf(q)[x], 
$ 
which is irreducible over $\gf(q)$. It then follows that $ x^{rn}-1=x^{rn}-\lambda^r=\prod_{i \in  \Gamma_{(q, rn)}} \m_{\beta^i}(x)$. Define 
$$\Gamma_{(q, rn,r)}^{(1)}=\{i: i \in \Gamma_{(q,rn)}, \, i \equiv 1 \pmod{r} \}.$$ 
Then $x^{n}-\lambda=\prod_{i \in  \Gamma_{(q,rn,r)}^{(1)}} \m_{\beta^i}(x)$.

\begin{lemma}\label{lem151} \cite{SWD22}
Let $n$ be a positive integer with $\gcd(q,n)=1$ and let $r$ be a positive divisor of $q-1$. If $\ord_n(q)=\ell$, then $\ord_{r n}(q)=\frac{r}{\gcd(\frac{q^\ell-1}{n},r)}\ell$, which is the size $\ell_1$ of $C_1^{(q, rn)}$, and the size $\ell_i$ of each $q$-cyclotomic coset $C_i^{(q, rn)}$ is a divisor of $\ord_{rn}(q)$.
\end{lemma}

\subsection{The trace representation of constacyclic codes} 

Constacyclic codes and their duals have the following relation, which is a fundamental result.

\begin{lemma}\cite{KS90}\label{lem-sdjoin2}
The dual code of an $[n,k]$ $\lambda$-constacyclic code $\mathcal{C}$ generated by $g(x)$ is an $[n,n-k]$ $\lambda^{-1}$-constacyclic code $\mathcal{C}^{\perp}$ generated by  $\widehat{h}(x)=h_0^{-1}x^kh(x^{-1})$, where $h(x)=(x^n-\lambda)/g(x)$ is the check polynomial of $\mathcal{C}$ and $h_0$ is the coefficient of $x^0$ in $h(x)$.
\end{lemma}

The trace representation of $\lambda$-constacyclic codes is documented below (see \cite{DY10}, \cite{SR2018}, \cite[Theorem 1]{SZW20}). 

\begin{lemma}\label{lem-01} 
Let $\lambda \in {\rm GF}(q)^*$ with $\ord(\lambda)=r$. Let $n$ be a positive integer such that $\gcd(n, q)=1$. Define $m=\ord_{rn}(q)$ and let $\beta \in  \gf(q^m)$ be a primitive $rn$-th root of unity such that $\beta^n=\lambda$. Let $\mathcal{C}$ be a $\lambda$-constacyclic code of length $n$ over $\gf(q)$. Suppose that $\mathcal{C}$ has altogether $s$ pairwise 
non-conjugate nonzeros, $\beta^{i_1}, \beta^{i_2}, \cdots, \beta^{i_s}$, which are $s$ roots of its check polynomial.
Then $\mathcal{C}$ has the trace representation $\mathcal{C}=\{\mathbf{c}(a_1,a_2,\ldots,a_s)\,:\,a_j\in {\rm GF}(q^{m_j}),\,1\leq j\leq s\}$, where
$$\mathbf{c}(a_1,a_2,\ldots,a_s)=\left(\sum_{j=1}^s{\rm Tr}_{q^{m_j}/q}(a_j\beta^{-ti_j})\right)_{t=0}^{n-1},$$
$m_j=|C_{i_j}^{(q, rn)}|$ and $C_{i_j}^{(q, rn)}$ is the $q$-cyclotomic coset of $i_j$ modulo $rn$.
\end{lemma}

Lemma \ref{lem-01} is very useful in determining the parameters and weight distributions of some constacyclic codes. We will make use of this lemma later in this paper.

\subsection{The BCH bound for constacyclic codes} 

The following lemma documents the BCH bound for constacyclic codes over finite fields, which is a generalization of the BCH bound of cyclic codes. 

\begin{lemma}\label{lem-BCHbound}\cite[Lemma 4]{KS90}[The BCH bound for constacyclic codes] 
Let $\mathcal{C}$ be a $\lambda$-constacyclic code of length $n$ over $\gf(q)$ and $\beta \in  \gf(q^m)$ be a primitive $rn$-th root of unity such that $\beta^n=\lambda$. Let $g(x)$ be the generator polynomial of $\mathcal{C}$. If there are integers $e$, $h$, $\delta$ with $\gcd(e,n)=1$ and $2\leq \delta \leq n$ such that
$$g(\beta^{1+reh})=g(\beta^{1+re(h+1)})=\cdots=g(\beta^{1+re(h+\delta-2)})=0,$$
then the minimum Hamming distance of $\mathcal{C}$ is at least $\delta$.
\end{lemma}

\subsection{Automorphism groups and equivalence of linear codes} 

Two linear codes $\C_1$ and $\C_2$ are said to be  {\em permutation-equivalent\index{permutation equivalent for codes}} if there is a permutation of coordinates which sends $\C_1$ to $\C_2$. This permutation could be described employing a \textit{permutation matrix}\index{permutation matrix}, which is a square matrix with exactly one 1 in each row and column and 0s elsewhere. The set of coordinate permutations that map a code $\C$ to itself forms a group, which is referred to as the \emph{permutation automorphism group\index{permutation automorphism group of codes}} of $\C$ and denoted by $\PAut(\C)$. 

A \emph{monomial matrix\index{monomial matrix}} over $\gf(q)$ is a square matrix having exactly one 
nonzero element of $\gf(q)$  in each row and column. A monomial matrix $M$ can be written either in 
the form $DP$ or the form $PD_1$, where $D$ and $D_1$ are diagonal matrices and $P$ is a permutation 
matrix. 

Let $\C_1$ and $\C_2$ be two linear codes of the same length over $\gf(q)$. Then $\C_1$ and $\C_2$ 
are said to be \emph{monomially-equivalent\index{monomially equivalent}} if there is a nomomial matrix over $\gf(q)$ 
such that $\C_2=\C_1M$. Monomial equivalence and permutation equivalence are precisely the same for 
binary codes. If $\C_1$ and $\C_2$ are monomially-equivalent, then they have the same weight distribution. 
The set of monomial matrices that map $\C$ to itself forms the group $\MAut(\C)$,  which is called the 
\emph{monomial automorphism group\index{monomial automorphism group}} of $\C$. By definition, we have 
$\PAut(\C) \subseteq \MAut(\C)$. Two linear codes $\C_1$ and $\C_2$ of the same length over $\gf(q)$ 
are said to be \emph{scalar-equivalent\index{scalar-equivalent}} if there is an invertible diagonal matrix $D$ 
over $\gf(q)$ such that $\C_2=\C_1D$. 

Two codes $\C_1$ and $\C_2$ are said to be \textit{equivalent}\index{equivalent} if there is a monomial 
matrix $M$ and an automorphism $\gamma$ of $\gf(q)$ such that $\C_1=\C_2 M \gamma$.  
All three are the same if the codes are binary; monomial equivalence and equivalence are the same if the field considered has a prime number of elements. 

The \textit{automorphism group}\index{automorphism group} of $\C$, denoted by $\GAut(\C)$, is the set 
of maps of the form $M\gamma$, 
where $M$ is a monomial matrix and $\gamma$ is a field automorphism, that map $\C$ to itself. In the binary 
case, $\PAut(\C)$,  $\MAut(\C)$ and $\GAut(\C)$ are the same. If $q$ is a prime, $\MAut(\C)$ and 
$\GAut(\C)$ are identical. In general, we have 
$$ 
\PAut(\C) \subseteq \MAut(\C) \subseteq \GAut(\C). 
$$ 
In this paper, we consider the monomial equivalence of linear codes. Two monomially-equivalent codes have the 
same parameters and weight distribution. 
If a linear code $\C$ is monomially-equivalent 
to a constacyclic code $\C_2$, we prefer $\C_2$ to $\C$ as constacyclic codes have a better algebraic structure than 
general linear codes.

\subsection{The projective Reed-Muller codes}\label{SEC-PRMCODE} 

 Let $m \geq 2$ be an integer. A point of the projective geometry $\PG(m-1, \gf(q))$ is given in homogeneous coordinates by $(x_1,x_2,\dots, x_{m})$ where all $x_i$ are in $\mathrm{GF}(q)$ and are not all zero. Each point of $\PG(m-1, \gf(q))$ has $q-1$ coordinate representations, as $(ax_1, ax_2,\dots,ax_{m})$ and $(x_1, x_2,\dots,x_{m})$ generate the same $1$-dimensional subspace of $\mathrm{GF}(q)^{m}$ for any nonzero $a\in \mathrm{GF}(q)$. 

Let $\gf(q)[x_1, x_2,\dots, x_m]$ be the set of polynomials in $m$ indeterminates over $\gf(q)$, which is a linear space over $\gf(q)$. Let $A(q, m, h)$ be the subspace of $\gf(q)[x_1, x_2,\dots, x_m]$ generated by all the homogeneous polynomials of degree $h$. Let $n=\frac{q^m-1}{q-1}$ and let $\{\mathbf{x}^1, \mathbf{x}^2, \cdots, \mathbf{x}^{n}\}$ be a set of projective points in $\mathrm{PG}(m-1,\gf(q))$. Then the \emph{$h$-th order projective Reed-Muller code}\index{projective Reed-Muller code} $\mathrm{PRM}(q, m, h)$ of length $n$ is defined by 
\begin{eqnarray*} 
\mathrm{PRM}(q,m,h)=\left \{\left (f(\mathbf{x}^1),f(\mathbf{x}^2), \dots, f(\mathbf{x}^n) \right ): f\in A(q, m, h) \right \}.
\end{eqnarray*} 
The code $\mathrm{PRM}(q,m,h)$ depends on the choice of the set  $\{\mathbf{x}^1,\mathbf{x}^2, \cdots, \mathbf{x}^{n}\}$ of coordinate representatives of the point set in $\mathrm{PG}(m-1,\gf(q))$, but is unique up to monomial equivalence (in fact, up to scalar equivalence). The parameters of $\mathrm{PRM}(q,m,h)$ are known and documented in the following theorem \cite{BR14,Lachaud,Sorensen}. 

\begin{theorem}\label{thm-PRMcode7} 
Let $m \geq 2$ and $1 \leq h \leq (m-1)(q-1)$. Then the linear code $\mathrm{PRM}(q,m,h)$ has length $n=\frac{q^m-1}{q-1}$ and minimum distance $(q-v)q^{m-2-u}$, where $h-1=u(q-1)+v$ and $0 \leq v <q-1$. Furthermore, 
\begin{eqnarray*}\label{eqn-PRMcode7}
\dim(\mathrm{PRM}(q,m,h))=\sum_{t \equiv h \pmod{q-1} \atop 0 < t \leq h} \left(  \sum_{j=0}^m (-1)^j \binom{m}{j} \binom{t-jq+m-1}{t-jq} \right).   
\end{eqnarray*}
\end{theorem}  

\begin{theorem}\label{prm-8}
Let $m\geq 2$ and $1\leq h\leq (m-1)(q-1)$. If $h\not \equiv 0\pmod{q-1}$, then $$\mathrm{PRM}(q, m, h)^{\bot}=\mathrm{PRM}(q, m, (m-1)(q-1)-h).$$
\end{theorem}

By Theorem \ref{thm-PRMcode7} and definition, $\mathrm{PRM}(q,m,1)$ 
is monomially-equavalent to the Simplex code. The weight distribution of $\mathrm{PRM}(q,m,2)$ was settled in \cite{Lisx19}. 
It was pointed out in \cite{BM01,Sorensen} that the code $\mathrm{PRM}(q,m,h)$ is not cyclic in general, but is equivalent to 
a cyclic code if $\gcd(m, q-1)=1$ or $h \equiv 0 \pmod{q-1}$.     
Later in this paper, we will compare some newly constructed constacyclic codes with the projective Reed-Muller codes.   
This explains why we introduced the projective Reed-Muller codes here. We will need the following theorem later. 

\begin{theorem}\label{thm-Lishuxing} \cite{Lisx19} 
Let $m \geq 2$. 
Then the weight distribution of $\mathrm{PRM}(q,m,2)$ is given by 
\begin{eqnarray*}
A_0 &=& 1, \\ 
A_{q^{m-1}} &=& q^{m}-1+ \sum_{j=1}^{\lfloor (m-1)/2 \rfloor} q^{j^2+j} \frac{\prod_{i=m-2j}^{m} (q^i-1)}{\prod_{i=1}^j (q^{2i}-1)}, \\
A_{q^{m-1}-\tau q^{m-1-j}} &=& \frac{q^{j^2}(q^j+\tau)}{2}  \frac{\prod_{i=m-2j+1}^{m} (q^i-1)}{\prod_{i=1}^j (q^{2i}-1)}, \ 
1 \leq j \leq \left \lfloor \frac{m}{2} \right\rfloor, \ \tau \in \{1, -1\}, 
\end{eqnarray*} 
and $A_h=0$ for other $h$. 
\end{theorem}

\subsection{The nonprimitive Reed-Muller codes}\label{sec-PGRMcode}

Let $m\geq 2$ be an integer  and let $r>1$ be a divisor of $q-1$. Let $\ell=(q-1)h+\ell_0<(q-1)m$, where $0\leq \ell_0\leq q-2$ and $\ell_0\equiv 0 \pmod{r}$. Let $\mathrm{P}(q, m, r, \ell)$ be the linear subspace of $\gf(q)[x_1,x_2, \dots, x_{m}]$, which is spanned by all monomials $x_1^{i_1}x_2^{i_2} \cdots x_{m}^{i_{m}}$ satisfying the following three conditions: 
\begin{enumerate}
\item $0\leq i_j\leq q-1$, for $1\leq j\leq m$,
\item $\sum_{j=1}^{m} i_j \equiv 0 \pmod{r}$,
\item $\sum_{j=1}^{m} i_j  \leq \ell $.
\end{enumerate} 

Let $\beta$ be a primitive element of $\gf(q^m)$ and let $\m_{\beta}(x)=\sum_{i=0}^{m-1}\epsilon_i x^i+x^m$, where $\epsilon_i\in \gf(q)$. Let
 $$\mathbf{M}=\begin{pmatrix}
0& 1& 0& \cdots &0\\
0& 0& 1& \cdots &0\\
\vdots &\vdots &\vdots &\ddots&\vdots \\
0& 0&0&\cdots& 1\\	
-\epsilon_0& -\epsilon_1 &-\epsilon_2&\cdots& -\epsilon_{m-1}	
\end{pmatrix}
 $$
be the companion matrix of $\m_{\beta}(x)$. Let $n=\frac{q^m-1}r$ and $\mathbf{e}=(1,0,\ldots,0)$. Then the \emph{nonprimitive generalized Reed-Muller code} $\mathrm{NGRM}(q, m, r, h)$ of length $n$ is defined by 
\begin{align*}
%\label{eq:PRM}
\mathrm{NGRM}(q,m,r, h)=\left \{\left (f(\mathbf{e}), f(\mathbf{e}\mathbf{M}),\dots, f(\mathbf{e}\mathbf{M}^{n-1}) \right ): f\in \mathrm{P}(q, m, r, \ell ) \right \}.
\end{align*}
In particular, when $r=q-1$, it is easily verified that $\{ \mathbf{e}, \mathbf{e}\mathbf{M}, \cdots, \mathbf{e}\mathbf{M}^{n-1}\}$ is the set of projective points in $\mathrm{PG}(m-1,\gf(q))$. Then the code $\mathrm{NGRM}(q, m, q-1, h)$ is also called the \emph{$h$-th order projective generalized Reed-Muller code} $\mathrm{PGRM}(q, m, h)$ of length $n$. 

\begin{theorem}\cite{DGM70}
Let $\ell=(q-1)h+\ell_0<(q-1)m$, where $0 \leq \ell_0 \leq q-2$ and $\ell_0\equiv 0 \pmod{r}$. Then the minimum weight of $\mathrm{NGRM}(q, m, r, h)$ is $\frac{(q-\ell_0)q^{m-h-1}-1}{r}$ and 
\begin{eqnarray}\label{eqn-dimPGRMc}
\dim(\mathrm{NGRM}(q, m, r, h))=\left| \left\{ 0 \leq j \leq \frac{q^m-1}{r}: \wt_q(j r) \leq \ell  \right\}   \right|.  
\end{eqnarray}
\end{theorem}

Note that the minimum distance of $\mathrm{NGRM}(q,m,r,h)$ is known to be $\frac{(q-\ell_0)q^{m-h-1}-1}{r}$. But the expression in (\ref{eqn-dimPGRMc}) is not specific, and no specific formula for $\dim(\mathrm{NGRM}(q, m, r, h))$ is known. Later we will compare the codes $\mathrm{NGRM}(q,m, r, h)$ with the constacyclic codes presented in this paper. To this end, we present the following example. 

\begin{example}\label{exam-PGRMc} 
The parameters of the codes $\mathrm{NGRM}(3,4,2,h)$ for $0 \leq h \leq 3$ are given below. 
\begin{eqnarray*}
[40, 1, 40], \  [40, 11, 13], \  [40, 30, 4], \  [40, 40, 1]. 
\end{eqnarray*}   
\end{example} 

\subsection{The punctured Dilix codes}\label{sec-DilixCode} 

In this subsection, we outline a type of cyclic codes,  called \emph{punctured Dilix codes}\index{punctured Dilix code} \cite{DLX18}. Let $N=q^m-1$, where $m$ is a positive integer. 
Let $\beta$ be a primitive element of $\gf(q^m)$. For any $1 \leq h  \leq m$, we define a polynomial 
\begin{eqnarray*}\label{eqn-generatorplym}
\omega_{(q,m,h )}(x)=\prod_{\myatop{1 \leq a \leq n-1}{1 \leq \wt(a) \leq h  }} (x-\beta^a).  
\end{eqnarray*}
Since $\wt(a)$ is a constant function on each $q$-cyclotomic coset modulo $N$, 
$\omega_{(q, m, h )}(x)$ is a polynomial over $\gf(q)$. By definition, $\omega_{(q, m, h )}(x)$ is a 
divisor of $x^N-1$. 
Let $\Omega{(q,m, h )}$ denote the cyclic code over 
$\gf(q)$ with length $N$ and generator polynomial $\omega_{(m,q,h )}(x)$. 

\begin{theorem}\label{thm-gRMcode} \cite{DLX18}
Let $m \geq 2$ and $1 \leq h  \leq m-1$. Then  
$\Omega{(q, m, h)}$ has parameters 
$$\left[q^m-1, q^m-\sum_{i=0}^{h } \binom{m}{i} (q-1)^i , d \geq \frac{q^{h+1}-1}{q-1} \right].$$ 
\end{theorem}

Later we will use the codes $\Omega(q,m,h)$ to construct some constacyclic codes. 
This explains why we introduced the punctured Dilix codes $\Omega(q,m,h)$ here. 

\section{A general construction of constacyclic codes of length $\frac{q^m-1}{r}$ with cyclic codes of length $q^m-1$} \label{sec-3}

In this section, we present a general construction of constacyclic codes of length $\frac{q^m-1}{r}$ with cyclic codes of length $q^m-1$ over $\gf(q)$, where $r$ is a positive divisor of $q-1$. Throughout this section, let $n=\frac{q^m-1}{r}$, where $m$ is an integer with $m \geq 2$. Define $N=rn=q^m-1$. Let $\beta$ be a primitive element of $\gf(q^m)$ and $\lambda=\beta^{n}$. Then $\lambda$ is an element of $\gf(q)^*$ with order $r$. 

Let $\C$ be a cyclic code of length $N$ over $\gf(q)$ with generator polynomial 
$$
g(x)=\prod_{i \in D(\C)} (x-\beta^i), 
$$ 
where $D(\C)$ is the union of some $q$-cyclotomic cosets modulo $N$ and is called the \emph{defining set} of $\C$ 
with respect to the primitive element $\beta$ of $\gf(q^m)$. 
Put 
\begin{eqnarray*}
\underline{D}(\C)=\{i \in D(\C): i \equiv 1 \pmod{r} \}. 
\end{eqnarray*} 
If $\underline{D}(\C) = \emptyset$, define $\underline{g}(x)=1$. 
If $\underline{D}(\C) \neq \emptyset$, define 
\begin{eqnarray*}
\underline{g}(x)=\prod_{i \in \underline{D}(\C)} (x-\beta^i). 
\end{eqnarray*} 
Then the following hold: 
\begin{enumerate}
\item $\underline{g}(x)$ is a polynomial over $\gf(q)$. 

\item $\underline{g}(x)=\gcd(g(x), x^n-\lambda)$.
\end{enumerate} 
Let $\underline{\C}$ denote the $\lambda$-constacyclic code of length $n$ over $\gf(q)$ with generator polynomial $\underline{g}(x)$. By definition, $\underline{\C}$ is constructed from the given cyclic code $\C$. In particular, the following hold:
\begin{enumerate}
\item If $(x^n-\lambda) \mid g(x)$, i.e., $\underline{D}(\C)=\Gamma_{(q,N,r)}^{(1)}$, then $\underline{\C}=\{ \0 \}$.	
\item If $\gcd(g(x), x^n-\lambda)=1$, i.e., $\underline{D}(\C)= \emptyset$, then $\underline{\C}=\gf(q)^n$.
\end{enumerate}
This general construction produces a nontrivial code only when $\underline{D}(\C) \not \in \{\emptyset, \Gamma_{(q, N,r)}^{(1)}\}$. 

By definition, 
$$
\dim(\C)=N-\deg(g)=N-|D(\C)|
$$ 
and 
$$
\dim(\underline{\C})=n-\deg(\underline{g})=n-|\underline{D}(\C)|. 
$$
Hence, it may not be easy to determine $\dim(\underline{\C})$ even if $\dim(\C)$ is known. However, this may be possible in some special cases. It is clear that $$x^{rn}-1=\prod_{i=0}^{r-1}(x^n-\lambda^i),$$ and $\gcd(x^n-\lambda^i, x^n-\lambda^j)=1$ for $i\neq j$. For a given $g(x)\mid (x^{N}-1)$, let $\underline{g}_i(x)=\gcd(g(x), x^n-\lambda^i)$. Then $\underline{g}_1(x)=\underline{g}(x)$. Let $\mathrm{Ind}(\C)=\{i: \underline{g}_i(x)\neq 1, 0\leq i\leq r-1 \}$, then $$g(x)=\prod_{i\in \mathrm{Ind}(\C)}\underline{g}_i(x).$$

\begin{theorem}\label{pro:13}
Let notation be the same as before. Assume that $\gcd(g(x),x^n-\lambda)\neq 1$ and $\gcd(g(x),x^n-\lambda)\neq x^n-\lambda$. Then the following hold:
\begin{enumerate}
\item $d(\C)\leq |\mathrm{Ind}(\C)|\cdot  d(\underline{\C})$. 
\item If $1\leq |\mathrm{Ind}(\C)|\leq r-1$, then $2\leq d(\C)\leq |\mathrm{Ind}(\C)|+1$.
\item The code $\underline{\C}=\{c(x) \pmod{x^n-\lambda}: c(x)\in \C \}$. \end{enumerate}
\end{theorem}
\begin{proof} 1. For any $\underline{c}(x)\in \underline{\C}$, we have 
$$c(x):=\underline{c}(x)\prod_{i\in \mathrm{Ind}(\C)\backslash \{1\}}(x^n-\lambda^i)\in \C.$$
Since $\prod_{i\in  \mathrm{Ind}(\C)\backslash \{1\}}(x^n-\lambda^i)$ can be expanded as a sum of the form $\sum a_ix^{ni}$, we have 
$$\wt(c(x))=\wt(\underline{c}(x))\cdot \wt\left(\prod_{i\in  \mathrm{Ind}(\C)\backslash \{1\}}(x^n-\lambda^i)\right).$$ 
Consequently,
$$d(\C)\leq \min\left\{ \wt\left(\underline{c}(x)\prod_{i\in  \mathrm{Ind}(\C)\backslash \{1\}}(x^n-\lambda^i)\right):  \0 \neq \underline{c}(x)\in \underline{\C}\right\}\leq |\mathrm{Ind}(\C)|\cdot d(\underline{\C}).$$

2. If $1\leq |\mathrm{Ind}(\C)|\leq r-1$, then $\0 \neq \prod_{i\in  \mathrm{Ind}(\C)}(x^n-\lambda^{i})\in \C$. Note that $$\wt\left(\prod_{i\in  \mathrm{Ind}(\C)}(x^n-\lambda^{i})\right)\leq |\mathrm{Ind}(\C)|+1,$$ we have $d(\C)\leq | \mathrm{Ind}(\C)|+1$.

3. Let $\mathrm{Res}(\C)=\{c(x) \pmod{x^n-\lambda}: \ c(x)\in \C  \}$. Let $c(x)\in \C$, then there is $\underline{c}(x) \in \gf(q)[x]/\langle x^n-\lambda \rangle$ such that $\underline{c}(x)\equiv c(x)\pmod{x^n-\lambda}$. Clearly, $$\gcd(c(x), x^n-\lambda)=\gcd(\underline{c}(x), x^n-\lambda).$$ 
  Then $\underline{g}(x)$ divides $\underline{c}(x)$. It follows that $\underline{c}(x)\in \underline{\C}$. Consequently, $\mathrm{Res}(\C)\subseteq \underline{\C}$.
  
  Let $\underline{c}(x)\in \underline{\C}$. It is easily verified that $\gcd\left(\frac{x^n-\lambda}{\underline{g}(x)}, \frac{g(x)}{\underline{g}(x)}\right)=1$. Then there are $a_1(x)$ and $a_2(x)$ such that $ a_1(x)\frac{x^n-\lambda}{\underline{g}(x)}+a_2(x)\frac{g(x)}{\underline{g}(x)}=1$. It follows that
  $$ a_2(x)\frac{g(x)}{\underline{g}(x)} \underline{c}(x)=\underline{c}(x)-a_1(x)\frac{x^n-\lambda}{\underline{g}(x)}\underline{c}(x).$$
  Note that $\underline{g}(x)\mid \underline{c}(x)$, we have $g(x)\mid \frac{g(x)}{\underline{g}(x)} \underline{c}(x)$ and $(x^n-\lambda) \mid \frac{x^n-\lambda}{\underline{g}(x)}\underline{c}(x)$. Therefore, 
  $$ c(x):=a_2(x)\frac{g(x)}{\underline{g}(x)} \underline{c}(x)\in \C$$
  and $c(x)\equiv \underline{c}(x) \pmod{x^n-\lambda}$. Consequently, $\underline{\C} \subseteq \mathrm{Res}(\C)$. The desired conclusion follows. 
\end{proof}

The third conclusion of Theorem \ref{pro:13} shows that there is no clear connection between $d(\underline{\C})$ and $d(\C)$ 
in general.

\begin{example}
	Let $(q, m, r)=(3, 4, 2)$. Then $n=\frac{q^m-1}2=40$ and $N=80$. Let $\beta$ be a primitive element of $\gf(3^4)$ with $\beta^4-\beta^3-1=0$.
	\begin{enumerate}
		\item Let $\C$ be the cyclic code of length $N$ over $\gf(q)$ with generator polynomial $g(x)=\m_{\beta}(x)$, then $\underline{\C}$ is the negacyclic code of length $n$ over $\gf(q)$ with generator polynomial $\underline{g}(x)=\m_{\beta}(x)$. Clearly, $\mathrm{Ind}(\C)=\{ 1\}$.  Then $\C$ has parameters $[80, 76, 2]$ and $\underline{\C}$ has parameters $[40,36 ,3]$. It is clear that $d(\C)< d(\underline{\C})$.
		\item Let $\C$ be the cyclic code of length $N$ over $\gf(q)$ with generator polynomial $g(x)=(x-1)\m_{\beta}(x)$, then $\underline{\C}$ is the negacyclic code of length $n$ over $\gf(q)$ with generator polynomial $\underline{g}(x)=\m_{\beta}(x)$. Clearly, $\mathrm{Ind}(\C)=\{0, 1\}$. Then $\C$ has parameters $[80, 75, 3]$ and $\underline{\C}$ has parameters $[40,36 ,3]$. It is clear that $d(\C)=d(\underline{\C})$.
		\item Let $\C$ be the cyclic code of length $N$ over $\gf(q)$ with generator polynomial $g(x)=(x^n-1)\m_{\beta}(x)$, then $\underline{\C}$ is the negacyclic code of length $n$ over $\gf(q)$ with generator polynomial $\underline{g}(x)=\m_{\beta}(x)$. Clearly, $\mathrm{Ind}(\C)=\{0, 1\}$. Then $\C$ has parameters $[80, 36, 6]$ and $\underline{\C}$ has parameters $[40,36 ,3]$. It is clear that $d(\C)=2d(\underline{\C})$.
	\end{enumerate}
\end{example}

Later in this paper, we will use this general construction to obtain two classes of $\lambda$-constacyclic codes of length $\frac{q^m-1}{r}$ over $\gf(q)$, where $r>1$ and $r\mid (q-1)$. 

%\begin{example} 
%Let $q>2$ be a prime power and let $m \geq 2$ and $r=q-1$. Let $\beta$ be a primitive element of $\gf(q^m)$ and $\lambda=\beta^{(q^m-1)/(q-1)}$. Let $\C$ denote the cyclic code of length $N=q^m-1$ with generator polynomial $g(x)=\m_{\beta}(x) \m_{\beta^{q+1}}(x)$. It is easily seen that 
%$$
%\underline{D}(\C)=C_{1}^{(q, N)}  
%$$ 
%and 
%$$
%\underline{g}(x)=\m_{\beta}(x). 
%$$ 
%Then  the $\lambda$-constacyclic code $\underline{\C}$ is the Hamming code and $\underline{\C}^\perp$ is the Simplex code.   
%\end{example} 

\section{The first class of constacyclic codes}\label{sec-1stclass} 

We follow the previous notation. Throughout this section, let $r>1$ and $r\mid (q-1)$. Let $n=\frac{q^m-1}{r}$, where $m$ is an integer with $m \geq 2$. Define $N=rn=q^m-1$. Then it follows from Lemma \ref{lem151} that 
$ 
\ord_n(q)=\ord_{N}(q)=m. 
$ 
Let $\Gamma_{(q,N)}$ be the set of $q$-cyclotomic coset leaders modulo $N$ and let 
$$
\Gamma_{(q,N,r)}^{(1)}=\{i: i \in \Gamma_{(q,N)}, \, i \equiv 1 \pmod{r} \}.
$$ 
Let $\beta$ be a primitive element of $\gf(q^m)$ and let $\lambda=\beta^{(q^m-1)/r}$. Then $\lambda\in  \gf(q)^*$ with $\ord(\lambda)=r$.  
Let $\ell$ be a positive integer with $1 \leq \ell \leq m$. Define 
\begin{eqnarray*}
g'_{(q, m, r,\ell)}(x) = \prod_{i \in \Gamma_{(q, N,r)}^{(1)}  \atop 1 \leq \wt(i) \le \ell} \m_{\beta^i}(x). 
\end{eqnarray*} 
Let 
\begin{eqnarray*}
D'_{(q,m, r,\ell)}=\bigcup_{i \in \Gamma_{(q, N,r)}^{(1)}  \atop 1 \leq \wt(i) \le \ell} C_i^{(q, N)}. 
\end{eqnarray*} 
Then $\{\beta^i: i \in D'_{(q, m, r, \ell)}\}$ is the set of all zeros of $g'_{(q, m, r,\ell)}(x)$. It is easily verified that $D'_{(q, m,r,\ell)}$ is invariant under the permutation $q y \mod N$ of $\Z_N$. Consequently, $g'_{(q, m, r,\ell)}(x)$ is over $\gf(q)$ and is a divisor of $x^n-\lambda$. Let $\C'(q, m, r,\ell)$ denote the $\lambda$-constacyclic code of length $n$ over $\gf(q)$ with generator polynomial $g'_{(q, m, r,\ell)}(x)$. By definition, $g'_{(q, m, r, m)}(x)=x^n-\lambda$ and the code $\C'(q, m, r, m)$ is the zero code and  $\C'(q, m, r, m)^\perp$ is the $[n, n, 1]$ code $\gf(q)^n$ over $\gf(q)$. Hence, we will consider the code $\C'(q, m, r,\ell)$ only for $1 \leq \ell \leq m-1$, and call $D'_{(q, m, r,\ell)}$ the \emph{defining set} of $\C'(q, m, r, \ell)$ with respect to the primitive element $\beta$ of $\gf(q^m)$.  
 
To  settle the dimension of this code $\C'(q, m, r, \ell)$, we need the following lemma. 

\begin{lemma}\label{lem-april91} 
Let $t$ be a positive integer and let $q$ be a prime power. Then the number of solutions $(x_1,x_2, \ldots, x_t)$ with 
$1 \leq x_i \leq q-1$ to the equation $x_1+x_2+\cdots+x_t \equiv 1 \pmod{r}$ is equal to $\frac{(q-1)^{t}}{r}$.  
\end{lemma} 

\begin{proof}
For any $(x_1,x_2,\ldots,x_{t-1})$ with $1\leq x_i\leq q-1$, let $a=x_1+x_2+\cdots+x_{t-1}$, then the equation $x_1+x_2+\cdots+x_t \equiv 1 \pmod{r}$ is equivalent to $x_t\equiv 1-a\pmod{r}$. For any $a$, the number of solutions $x_t$ with $1\leq x_t\leq q-1$ to the equation $x_t \equiv 1-a \pmod{r}$ is equal to $\frac{q-1}r$. The desired conclusion follows. 
\end{proof}

\begin{theorem}\label{thm-april92}
Let $1 \leq \ell \leq m-1$. Then 
$$
\dim(\C'(q,m, r,\ell))=\frac{q^m - \sum_{i=0}^\ell \binom{m}{i} (q-1)^i}{r} 
$$ 
and 
\begin{eqnarray}\label{eqn-distbound2}
 d(\C'(q,m, r, \ell)) \geq    \left\lfloor \frac{q^{\ell +1} -1 - 2(q-1) }{r  (q-1)} \right\rfloor +2.
\end{eqnarray}    
\end{theorem} 

\begin{proof}
Let $i$ be an integer with $1 \leq i \leq q^m-2$. Let the $q$-adic expression of $i$ be 
$$ 
i=\sum_{j=0}^{m-1} i_j q^j, \ \ 0 \leq i_j \leq q-1. 
$$ 
Then $i \equiv \sum_{j=0}^{m-1} i_j  \pmod{r} $. It then follows from Lemma \ref{lem-april91} that the number of $i$ with $1 \leq i \leq q^m-2$ such that $\wt(i)=t$ and $i \equiv 1 \pmod{r}$ is $\binom{m}{t}\frac{(q-1)^{t}}{r}$. Consequently, 
$$
\deg(g'_{(q, m, r,\ell)}(x))=\sum_{i=1}^\ell \binom{m}{i} \frac{(q-1)^{i}}{r}. 
$$
Thus, 
\begin{align*}
\dim(\C'(q, m, r, \ell))&= \frac{q^m-1}{r}-\sum_{i=1}^\ell \binom{m}{i} \frac{(q-1)^i}{r} \\
&=\frac{q^m - \sum_{i=0}^\ell \binom{m}{i} (q-1)^i}{r}. 	
\end{align*}

We now prove the lower bound on the minimum distance of the code $\C'(q, m, r,\ell)$. It is straightforward to verify that every integer $a$ with $1 \leq a \leq \frac{q^{\ell+1}-1}{q-1}-1$ has Hamming weight $\wt(a) \leq \ell$. It then follows from the definition of the code $\C'(q, m, r, \ell)$ that $\beta^i$ is a zero of $\C'(q, m, r, \ell)$ for each $i$ in the set 
$$  
\left\{1+rj: 0 \leq j \leq \left\lfloor \frac{q^{\ell +1} -1 - 2(q-1) }{r(q-1)} \right\rfloor    \right\}. 
$$ 
The desired lower bound then follows from the BCH bound for constacyclic codes (see Lemma \ref{lem-BCHbound}). 
\end{proof} 

Next we study the dual code of the constacyclic code $\C'(q, m, r, \ell)$. We have the following theorem.
 
% \begin{theorem} \label{CHTbound} \cite[Corollary 3]{RZ09} Let $\mathcal{C}$ be a $\lambda$-constacyclic code of length $n$ over $\gf(q)$ and $\beta \in  \gf(q^m)$ be a primitive $rn$-th root of unity such that $\beta^n=\lambda$. Let $g(x)$ be the generator polynomial of $\mathcal{C}$.  Assume that there are integers $\delta$, $t$, $b$, $c_1$ and $c_2$ where $\delta\geq 2$, $t\geq 0$, $b\geq 0$, $\gcd(n, c_1)=1$ and $\gcd(n, c_2)<\delta$ such that
%$$g(\beta^{1+r[b+i_1c_1+i_2 c_2]})=0,~0\leq i_1\leq \delta-2,~0\leq i_2\leq t.$$
%Then the minimum Hamming distance of $\mathcal{C}$ is at least $\delta+t$.
% \end{theorem}

\begin{theorem}\label{thm-april101}
Let $q\geq 3$, $r>1$ and $r\mid (q-1)$. Let $1 \leq \ell \leq m-1$. Then 
$$
\dim(\C'(q, m, r, \ell)^\perp)= \frac{ \sum_{i=1}^\ell \binom{m}{i}(q-1)^i }{r}
$$ 
and 
\begin{eqnarray}\label{eqn-ddistbound2}
d(\C'(q, m, r, \ell)^\perp) \geq q^{m-\ell}.  
\end{eqnarray} 
\end{theorem}  

\begin{proof}
The desired dimension of  $\C'(q,m, r, \ell)^\perp$ follows from the dimension of  $\C'(q, m, r, \ell)$. 
Note that $(\beta^{-1})^n = \lambda^{-1}$ and $\beta^{-1}$ is a primitive element of $\gf(q^m)$. By Lemma \ref{lem-sdjoin2}, 
the dual code $\C'(q, m, r, \ell)^\perp$ is a $\lambda^{-1}$-constacyclic code of length $n=\frac{q^m-1}{r}$ over $\gf(q)$ with generator polynomial 
$$
\prod_{i \in \Gamma_{(q, N,r)}^{(1)}  \atop \wt(i) > \ell} \m_{(\beta^{-1})^i}(x). 
$$ 
Let 
$$
D(q, m, r, \ell)=\{i \in \Z_N: \wt(i) \geq \ell +1, \ i \equiv 1 \pmod{r}\},
$$
then
$$
\bigcup_{i \in \Gamma_{(q, N,r)}^{(1)} \atop \wt(i) > \ell} C_i^{(q,N)}=D(q, m, r, \ell). 
$$
%The largest $i$ in $\Z_N$ such that $i \equiv 1 \pmod{r}$ and $\wt(i)=\ell$ is 
%$$
%(q-r)q^{m-\ell }+(q-1)[q^{m-\ell+1}+\cdots+q^{m-1}]=1+r\left( \frac{q^{m}-1}{r}-q^{m-\ell}  \right). 
%$$ 
Let 
$$
B:=\left\{ 1+r j:  \frac{q^{m}-1}{r}-q^{m-\ell}+1 \leq j \leq \frac{q^m-1}{r}-1  \right\}. 
$$
It is easily checked that $\wt(i) \geq \ell+1$ and $i \equiv 1 \pmod{r}$ for all $i \in B$. Hence, $B$ is a subset of $D(q, m, r, \ell)$. Consequently, $(\beta^{-1})^i$ is a zero of $\C'(q, m, r, \ell)^\perp$ for each $i \in B$. Note that 
$$
\left(\frac{q^m-1}{r}-1\right) -\left(\frac{q^m-1}r-q^{m-\ell}+1 \right)+1=q^{m-\ell}-1. 
$$ 
The desired conclusion then follows from the BCH bound for constacyclic codes (see Lemma \ref{lem-BCHbound}). 
\end{proof}

An interesting fact about the family of constacyclic codes $\C'(q, m, r, \ell)$ is the following.

\begin{corollary}\label{cor-1}
Let $m\geq 2$ and $r=q-1$. Then the constacyclic code $\C'(q, m, r, 1)$ has parameters 
	$$\left[\frac{q^m-1}{q-1}, \frac{q^m-1}{q-1}-m, 3 \right] $$
	and is monomially-equivalent to the Hamming code.  In addition, $\C'(q, m, r, 1)^\perp$ has parameters 
	$$\left[\frac{q^m-1}{q-1}, m, q^{m-1} \right] $$ 
	and is monomially-equivalent to the Simplex code.
\end{corollary}

\begin{proof}
 The desired dimension of the code $\C'(q, m, q-1, 1)$ follows from Theorem \ref{thm-april92}. It follows from Lemma \ref{lem-SPB1} that $d(\C'(q, m, q-1, 1)) \leq 4$. It then follows from Lemma \ref{lem-SPB2} that $$d(\C'(q,m,q-1,1)) \neq 4.$$ Again by Theorem \ref{thm-april92}, $d(\C'(q, m, q-1, 1)) \geq 3$. Consequently, $d(\C'(q, m, q-1, 1))=3$. Hence, $\C'(q, m, q-1,1)$ has the same parameters as the Hamming code of length $\frac{q^m-1}{q-1}$ over $\gf(q)$. It is well known that all linear codes over $\gf(q)$ with  parameters $$\left[\frac{q^m-1}{q-1}, \frac{q^m-1}{q-1}-m, 3\right]$$ are unique up to monomial equivalence. Therefore, $\C'(q, m, q-1, 1)$ is monomially-equivalent to the Hamming code and $\C'(q, m, q-1, 1)^\perp$ is monomially-equivalent to the Simplex code.  
\end{proof}

\begin{corollary}\label{cor-2} 
Let $m\geq 2$. Let $q$ be an odd prime power and $r=\frac{q-1}2$. Then the constacyclic code $\C'(q, m, r, 1)$ has parameters 
	$$\left[2\left(\frac{q^m-1}{q-1}\right), 2\left(\frac{q^m-1}{q-1}\right)-2m, 4 \right] $$
	and is distance-optimal. The dual code $\C'(q, m, r, 1)^{\bot}$ has the following results:
	\begin{itemize}
	\item When $m\geq 3$ is odd and $q$ is an odd prime, $\C'(q, m, r, 1)^{\bot}$ has parameters $$\left[2\left(\frac{q^m-1}{q-1}\right), 2m, 2q^{m-1}-q^{\frac{m-1}2} \right],$$
	and the weight distribution of $\C'(q, m, r, 1)^{\bot}$ is given in Table \ref{tab-ding1}.
	\item When $m\geq 2$ is even and $q$ is an odd prime, $\C'(q, m, r, 1)^{\bot}$ has parameters $$\left[2\left(\frac{q^m-1}{q-1}\right), 2m, 2q^{m-1}-(q-1)q^{\frac{m-2}2} \right],$$
	and the weight distribution of $\C'(q, m, r, 1)^{\bot}$ is given in Table \ref{tab-ding2}. 	
	\end{itemize} 
\end{corollary} 

\begin{proof}
The desired dimension of the code $\C'(q, m, r, 1)$ follows from Theorem \ref{thm-april92}. It follows from Lemma \ref{lem-SPB1} that $d(\C'(q, m, r, 1)) \leq 4$. Again by Theorem \ref{thm-april92}, $d(\C'(q, m, r, 1)) \geq 4$. Consequently, $d(\C'(q, m, r, 1))=4$, and $\C'(q, m, r, 1)$ is distance-optimal.

  It is easily checked that the generator polynomial of $\C'(q, m, r, 1)$ is $\m_{\beta}(x)\m_{\beta^{(q+1)/2}}(x)$. By Lemma \ref{lem-01}, the code $\C'(q, m, r, 1)^{\bot}$ has the trace representation
    $$\C'(q, m, r, 1)^{\bot}=\{\bc(a_1,a_2)=(\tr_{q^m/q}(a_1 \beta^{i}+a_2 \beta^{(\frac{q+1}2)i}) )_{i=0}^{n-1}: a_1, a_2\in \gf(q^m) \}. $$
    Define 
    $$\mathcal{EC}'(q, m, r, 1)^{\bot}=\{\widetilde{\bc}(a_1,a_2)=(\tr_{q^m/q}(a_1 \beta^{i}+a_2 \beta^{(\frac{q+1}2)i}) )_{i=0}^{q^m-2}: a_1, a_2\in \gf(q^m) \}. $$
    It is easily verified that $\mathcal{EC}'(q, m, r, 1)^{\bot}$ is the cyclic code of length $q^m-1$ over $\gf(q)$ with check polynomial $\m_{\beta^{-1}}(x) \m_{\beta^{-(q+1)/2}}(x)$. For each $(a_1,a_2)\in \gf(q^m)^2$, we have
    $$\widetilde{\bc}(a_1,a_2)=(\bc(a_1,a_2) \| \lambda \cdot \bc(a_1,a_2) \| \cdots \| \lambda^{\frac{q-3}2} \cdot \bc(a_1,a_2)), $$
where $\lambda=\beta^n\in \gf(q)^*$ and $\|$ denotes the concatenation of vectors. It follows that the constacyclic code $\C'(q, m, r, 1)^{\bot}$ has weight distribution $W(z)$ if and only if the cyclic code $\mathcal{EC}'(q, m, r, 1)^{\bot}$ has weight distribution $W(z^{\frac{q-1}2})$. When $q$ is an odd prime, the weight distribution of $\mathcal{EC}'(q, m, r, 1)^{\bot}$ was determined in \cite{LF08}. The desired result follows.  
\end{proof}

\begin{table*}[h]
	\begin{center}
\caption{Weight distribution of the code $\C'(q, m, r, 1)^{\bot}$ for odd $m$ }\label{tab-ding1}
\begin{tabular}{cc} \hline
  Weight $\omega$ & No. of codewords $A_\omega$ \\ \hline
  $0$& $1$\\ \hline
  $2q^{m-1}-q^{\frac{m-1}2} $& $q^{\frac{m-1}2}(q^{\frac{m-1}2}+1)(q^m-1)$\\ \hline
  $2 q^{m-1} $& $(q^m-1)(q^m-2q^{m-1}+1)$\\ \hline
  $2q^{m-1}+q^{\frac{m-1}2}$& $q^{\frac{m-1}2}(q^{\frac{m-1}2}-1)(q^m-1)$\\ \hline
\end{tabular}
\end{center}
\end{table*}

\begin{table*}[h]
	\begin{center}
\caption{Weight distribution of the code $\C'(q, m, r, 1)^{\bot}$ for even $m$ }\label{tab-ding2}
\begin{tabular}{cc} \hline
  Weight $\omega$ & No. of codewords $A_\omega$ \\ \hline
  $0$& $1$\\ \hline
  $2q^{m-1}-(q-1)q^{\frac{m-2}2} $& $\frac{(q^{\frac{m-2}2}+1)(q^{\frac{m}2}-1)(q^m-1)}{q^2-1}$\\ \hline
  $2q^{m-1}-2q^{\frac{m-2}2} $& $\frac{(q^\frac{m}2+1)^2(q-1)(q^m-1)}{4(q+1)}$\\ \hline
  $2q^{m-1}-q^{\frac{m-2}2}$& $q^{\frac{m-2}2}(q^{\frac{m}2}+1)(q^m-1)$\\ \hline
   $2q^{m-1}$& $\frac{(q^{m+1}-3q^m+q+1)(q^m-1)}{2(q-1)}$\\ \hline
   $2q^{m-1}+q^{\frac{m-2}2}$& $q^{\frac{m-2}2}(q^{\frac{m}2}-1)(q^m-1)$\\ \hline
   $2q^{m-1}+2q^{\frac{m-2}2}$& $\frac{(q^{\frac{m}2}-1)^2(q-1)(q^m-1)}{4(q+1)}$\\ \hline
    $2q^{m-1}+(q-1)q^{\frac{m-2}2}$& $\frac{(q^{\frac{m-2}2}-1)(q^{\frac{m}2}+1)(q^m-1)}{q^2-1}$\\ \hline
\end{tabular}
\end{center}
\end{table*}

\begin{example}
	Let $(q, m, r, \ell)=(5,2, 2,1)$. Let $\beta$ be the primitive element of $\gf(5^2)$ with $\beta^2+4\beta+2=0$. Then $\C'(5,2,2,1)$ has parameters $[12, 8, 4]$ and is distance-optimal. The dual code $\C'(5,2,2,1)^{\bot}$ has parameters $[12, 4, 6]$ and weight enumerator $1+8z^6+144z^8+144z^9+168z^{10}+96z^{11}+64z^{12}$.
\end{example}

\begin{example}
	Let $(q, m, r, \ell)=(5,3, 2,1)$. Let $\beta$ be the primitive element of $\gf(5^3)$ with $\beta^3+3\beta+3=0$. Then $\C'(5,3,2,1)$ has parameters $[62, 56, 4]$ and is distance-optimal. The dual code $\C'(5,3,2,1)^{\bot}$ has parameters $[62, 6, 45]$ and weight enumerator 
	$$1+3720z^{45}+9424z^{50}+2480z^{55}.$$
	Moreover, the code $\C'(5,3,2,1)^{\bot}$ has the best parameters known \cite{Grassl}. 
\end{example}

\begin{corollary}\label{cor-3} 
Let $m\geq 2$. Let $q$ be a prime power with $q\equiv 1\pmod{3}$, and let $r=\frac{q-1}3>1$. Then the constacyclic code $\C'(q, m, r, 1)$ has parameters 
	$$\left[3\left(\frac{q^m-1}{q-1}\right), 3\left(\frac{q^m-1}{q-1}\right)-3m, 5\leq d\leq 6 \right].$$
	\end{corollary} 
	
\begin{proof}
	The desired dimension of the code $\C'(q, m, r, 1)$ follows from Theorem \ref{thm-april92}. It follows from Lemma \ref{lem-SPB1} that $d(\C'(q, m, r, 1)) \leq  6$. Again by Theorem \ref{thm-april92}, $d(\C'(q, m, r, 1)) \geq 5$. The desired result follows.
\end{proof}
	
\begin{example}
Let $(q, m, r, \ell)=(7,2, 2, 1)$. Let $\beta$ be the primitive element of $\gf(7^2)$ with $\beta^2+6\beta+3=0$. Then the constacyclic code $\C'(7,2,2,1)$ has parameters $[24, 18, 5]$ and has the best parameters known \cite{Grassl}. 
\end{example}

Let $\Omega(q, m, \ell)$ denote the punctured Dilix code constructed in \cite{DLX18} (see also Section \ref{sec-DilixCode}).  Theorem  \ref{thm-april92} tells us that 
$$
\dim(\Omega(q, m, r, \ell))=r \cdot \dim(\C'(q, m, r, \ell)).  
$$
Experimental data indicates that the lower bound in (\ref{eqn-distbound2}) is good  in general. 
But the following problem is worth of investigation. 

\begin{open} 
Determine the minimum distance of $\C'(q, m, r, \ell)$ or improve the lower  in (\ref{eqn-distbound2})
for $2 \leq \ell \leq m-1$. 
\end{open} 

Experimental data shows that the lower bound in (\ref{eqn-ddistbound2}) is quite away from the true minimum distance. 

\begin{open} 
Determine the minimum distance of $\C'(q, m, r, \ell)^\perp$ or improve the lower bound in  (\ref{eqn-ddistbound2}) for $2 \leq \ell \leq m-1$. 
\end{open} 

\begin{example} 
Let $(q, m, r, \ell)=(3, 4, 2, 1)$. Let $\beta$ be the primitive element of $\gf(3^4)$ with $\beta^4 + 2\beta^3 + 2=0$. Then the code $\C'(3, 4, 2, 1)$ has parameters $[40, 36, 3]$ and $\C'(3, 4, 2, 1)^\perp$ has parameters $[40, 4, 27]$. The former is a perfect code and the latter meets the Griesmer bound. 
\end{example} 

\begin{example}\label{exam-342}  
Let $(q, m, r, \ell)=(3, 4, 2, 2)$. Let $\beta$ be the primitive element of $\gf(3^4)$ with $\beta^4 + 2\beta^3 + 2=0$. Then the code $\C'(3, 4, 2, 2)$ has parameters $[40, 24, 8]$ and $\C'(3, 4, 2, 2)^\perp$ has parameters $[40, 16, 12]$. The best ternary code known of length $40$ and dimension $24$ has minimum distance $9$ \cite{Grassl}.   
\end{example} 

\begin{example}\label{exam-343} 
Let $(q, m, r, \ell)=(3, 4, 2, 3)$. Let $\beta$ be the primitive element of $\gf(3^4)$ with $\beta^4 + 2\beta^3 + 2=0$. Then the code $\C'(3, 4, 2, 3)$ has parameters $[40, 8, 21]$  and has the best parameters known \cite{Grassl}, and $\C'(3, 4, 2, 3)^\perp$ has parameters $[40, 32, 4]$.  
\end{example} 

\begin{example} 
Let $(q, m, r, \ell)=(4, 3, 3, 2)$. Let $\beta$ be the primitive element of $\gf(4^3)$ with $\beta^6 + \beta^4 + \beta^3 + \beta + 1=0$. Then the code $\C'(4, 3, 3, 2)$ has parameters $[21, 9, 8]$ and $\C'(4, 3, 3, 2)^\perp$ has parameters $[21, 12, 6]$.  
\end{example} 

%\begin{example} 
%Let $(q, m, r, \ell)=(5, 3, 4, 2)$. Let $\beta$ be the primitive element of $\gf(5^3)$ with $\beta^3 + 3\beta+3=0$. Then the code $%\C'(5, 3, 4, 2)$ has parameters $[31, 16, 10]$ and has the best parameters known \cite{Grassl}, and $\C'(5, 3, 4, 2)^\perp$ has parameters $[31, 15, 11]$. The best quinary code known of length $31$ and dimension $15$ has minimum distance $12$ \cite{Grassl}.
%\end{example}

%\begin{example} 
%Let $(q, m, r, \ell)=(5, 3, 2, 2)$. Let $\beta$ be the primitive element of $\gf(5^3)$ with $\beta^3 + 3\beta+3=0$. Then the code $\C'(5, 3, 2, 2)$ has parameters $[62, 32, 16]$. The best quinary code known of length $62$ and dimension $32$ has minimum distance $17$ \cite{Grassl}.The dual code $\C'(5, 3, 2, 2)^\perp$ has parameters $[62, 30, 17]$.
%\end{example}  

%\begin{example} 
%Let $(q, m, r, \ell)=(7, 3, 6, 2)$. Let $\beta$ be the primitive element of $\gf(7^3)$ with $\beta^3+6\beta^2+4=0$. Then the code $\C'(7, 3, 6, 2)$ has parameters $[57, 36, 12]$  and has the best parameters known \cite{Grassl}, and $\C'(7, 3, 6, 2)^\perp$ has parameters $[57, 21, 19]$.  
%\end{example} 

The forgoing examples demonstrate that the code $\C'(q, m, r, \ell)$ and its dual $\C'(q, m, r, \ell)^\perp$ may be optimal or have the best parameters known sometimes. Below we explain some connection and difference among the code $\C'(q, m, q-1, \ell)$, the projective Reed-Muller codes and the nonprimitive generalized Reed-Muller codes. 

By Corollary \ref{cor-1}, $\C'(q, m, q-1, 1)^\perp$ is monomially-equivalent to $\PRM(q, m,1)$, as both codes are monomially-equivalent to the Simplex code. This is one connection between the codes $\C'(q, m,q-1,\ell)$ and the projective Reed-Muller codes. Consider now all the projective codes $\PRM(3,4, \ell)$ for all $\ell$ with $1 \leq \ell \leq 6$. It follows from Theorem \ref{thm-PRMcode7} that 
\begin{eqnarray*}
d(\PRM(3,4,1)) &=& 27, \\
d(\PRM(3,4,2)) &=& 18, \\
d(\PRM(3,4,3)) &=& 9, \\
d(\PRM(3,4,4)) &=& 6, \\ 
d(\PRM(3,4,5)) &=& 3, \\
d(\PRM(3,4,6)) &=& 2. 
\end{eqnarray*} 
By Example \ref{exam-342}, $d(\C'(3,4,2,2))=8$ and  $d(\C'(3,4,2,2)^\perp)=12$. This means that both $\C'(3,4,2,2)$ and  $\C'(3,4,2,2)^\perp$ cannot be monomially-equivalent to a code $\PRM(3, 4, \ell)$ for all $\ell$ with $1 \leq \ell \leq 6$. Hence, the two families of codes $\C'(q, m, q-1, \ell)$ and $\PRM(q, m, \ell)$ are different in general. Notice that $\C'(q, m, q-1, \ell)$ and the punctured Dilix code $\Omega(q, m, \ell)$ are not monomially-equivalent when $q>2$, as they have different lengths.

Compared with parameters of the codes $\textrm{NGRM}(3, 4, 2, \ell)$ in Example \ref{exam-PGRMc}, both $\C'(3,4,2, 2)$ and $\C'(3,4,2, 2)^\perp$ cannot be monomially-equivalent to a code $\textrm{NGRM}(3, 4,2,\ell)$ for all $\ell$ with $0 \leq \ell \leq 3$. Hence, the class of codes $\C'(q, m, q-1, \ell)$ and the class of codes  $\textrm{NGRM}(q, m, q-1, \ell)$ are different.   

\section{The second class of constacyclic codes}\label{sec-2ndclass} 

We follow the previous notation. Throughout this section, let $r>1$ and $r\mid(q-1)$. Let $n=\frac{q^m-1}{r}$, where $m$ is an integer with $m \geq 2$. Define $N=rn=q^m-1$, then it follows from Lemma \ref{lem151} that  
$ 
\ord_n(q)=\ord_{N}(q)=m. 
$ 
Let $\Gamma_{(q,N)}$ be the set of $q$-cyclotomic coset leaders modulo $N$ and let 
$$
\Gamma_{(q,N,r)}^{(1)}=\{i: i \in \Gamma_{(q,N)}, \, i \equiv 1 \pmod{r} \}.
$$ 
Recall that $r\mid(q-1)$, we have $\wt_q(i)\equiv i \pmod{r}$. Then $\wt_q(i)\equiv 1\pmod{r}$ for $i\in \Gamma_{(q,N,r)}^{(1)}$.

\subsection{Definition and basic properties of the constacyclic codes} 

Let $\beta$ be a primitive element of $\gf(q^m)$ and let $\lambda=\beta^{(q^m-1)/r}$. Then $\lambda \in \gf(q)^*$ with $\ord(\lambda)=r$. Let $\ell$ be an integer with $0 \leq \ell < (q-1)m-1$. Define 
\begin{eqnarray*}
g_{(q, m, r, \ell)}(x) = \prod_{i \in \Gamma_{(q, N,r)}^{(1)}  \atop \wt_q(i) <(q-1)m-\ell} \m_{\beta^i}(x). 
\end{eqnarray*}
Let 
\begin{eqnarray*}
D_{(q,m, r, \ell)}=\bigcup_{i \in \Gamma_{(q, N,r)}^{(1)}  \atop \wt_q(i) <(q-1)m-\ell  } C_i^{(q, N)}. 
\end{eqnarray*} 
Note that $\wt_q(i)\equiv i \pmod {r}$. It is easily checked that
$$D_{(q, m, r, \ell)}=\{ i\in \Z_N: \wt_q(i)<(q-1)m-\ell, \ \wt_q(i)\equiv 1\pmod{r} \}.$$
By definition, $\{\beta^i: i \in D_{(q, m, r, \ell)}\}$ is the set of all zeros of $g_{(q, m, r,\ell)}(x)$. It is easily verified that $D_{(q, m, r, \ell)}$ is invariant under the permutation $qy \mod N$ of $\Z_N$. Consequently, $g_{(q, m, r, \ell)}(x)$ is over $\gf(q)$ and is a divisor of $x^n-\lambda$. Let $\C(q, m, r, \ell)$ denote the $\lambda$-constacyclic code of length $n$ over $\gf(q)$ with generator polynomial $g_{(q, m, r, \ell)}(x)$. We call $D_{(q, m, r, \ell)}$ the \emph{defining set} of $\C(q, m, r, \ell)$ with respect to the primitive element $\beta$ of $\gf(q^m)$.  

\begin{theorem}\label{Thm-1}
Let $0 \leq \ell= r \ell_1+\ell_0 < m(q-1)-1$, where $0\leq \ell_0\leq r-1$. If $\ell_1=0$ and $0\leq \ell_0\leq r-2$, then $\C(q,m, r, \ell)=\{ \bzero\}$. Otherwise, $\C(q,m, r, \ell)=\C(q,m, r, r \ell_2+r-1)$, where 
\begin{align*}
	\ell_2=\begin{cases}
	\ell_1~& {\rm if~} \ell_0= r-1,\\
	\ell_1-1~&{ \rm if~} \ell_0\neq r-1.
\end{cases}
\end{align*}
\end{theorem}

\begin{proof}
Since $r\mid (q-1)$, we have $(q-1)m-\ell \equiv 1 \pmod r$ if and only if $\ell \equiv r-1 \pmod r$. If $\ell_1=0$ and $0\leq \ell_0\leq r-2$, i.e., $0\leq \ell\leq r-2$. Then $(q-1)m-\ell \leq (q-1)m-r+2$. In this case, $D_{(q, m, r, \ell)}=\{i\in \Z_N: \wt_q(i)\equiv 1\pmod{r} \}$, i.e., $g_{(q, m, r,\ell )}(x)=x^n-\lambda$. Consequently, $\C(q, m, r, \ell)=\{ \bzero\}$. If $\ell\geq r-1$ and $\ell_0<r-1$, then $\wt_q(i)<(q-1)m-\ell$ with $\wt_q(i)\equiv 1\pmod{r}$ if and only if 
 	$$\wt_q(i)\leq (q-1)m-(r\ell_1+r-1)<(q-1)m-(r\ell_1+r-1)+r$$
 	with $\wt_q(i)\equiv 1\pmod{r}$. The desired conclusion follows. 
\end{proof}
 
 It follows from Theorem \ref{Thm-1} that the class of $\lambda$-constacyclic codes $\C(q, m, r, \ell)$ contains only the following distinct codes 
$$
\C(q,m, r, r \ell_1 +r-1), \ 0 \leq \ell_1 \leq \left(\frac{q-1}r\right)m-2.  
$$ 

To determine the dimension of the $\lambda$-constacyclic code $\C(q, m, r, \ell)$, we need the following lemma.

\begin{lemma}\cite{Sorensen}\label{Sun-1}
	The number of ways one can place $t$ objects in $m$ cells such that no cell contains more than $s$ objects is 
	\begin{equation*}
		N(t, m, s)=\sum_{j=0}^m(-1)^j\binom{m}{j}\binom{t-j(s+1)+m-1}{t-j(s+1)}.
	\end{equation*}
\end{lemma}

The dimension of the $\lambda$-constacyclic code $\C(q, m, r, \ell)$ is documented in the next theorem.  

\begin{theorem}\label{thm-mycode5221} 
Let $\ell= r \ell_1+r-1$, where $0\leq \ell_1\leq (\frac{q-1}r) m-2$. Then 
\begin{align}\label{eqn-mycode5221}
\dim(\C(q,m, r,\ell))=&\sum_{t \equiv r-1 \pmod{r} \atop 0 < t \leq \ell} \left(  \sum_{j=0}^m (-1)^j \binom{m}{j} \binom{t-jq+m-1}{t-jq} \right) \notag \\ 
=& \sum_{t=0}^{\ell_1}  \sum_{j=0}^m (-1)^j \binom{m}{j}  \binom{tr+ r-1-jq+m-1}{tr+ r-1-jq}. 
\end{align}
\end{theorem}  

\begin{proof}
Define 
$$H(q, m, r, \ell)=\{i \in \Z_N: \wt_q(i) \geq (q-1)m-\ell, ~ \wt_q(i) \equiv 1 \pmod{r}\}.$$
    By definition, $\dim(\C(q,m, r, \ell)) =|H(q,m, r, \ell)|$. We now determine $|H(q,m, r, \ell)|$. 

For each $i \in \Z_N$, $i \equiv 1 \pmod{ r}$ if and only if $N-i \equiv r-1 \pmod{r }$. 
Furthermore, 
$$
\wt_q(N-i)=(q-1)m - \wt_q(i). 
$$ 
Consequently, $\wt_q(i) \geq (q-1)m-\ell$ if and only if $\wt_q(N-i) \leq \ell$. We then deduce that 
\begin{eqnarray}\label{eqn-may241}
|H(q,m, r, \ell)|=|\{i \in \Z_N: \wt_q(i) \leq \ell, \, \wt_q(i) \equiv r-1 \pmod{r}\}|. 
\end{eqnarray}

 From Lemma \ref{Sun-1}, the number of ways of picking $t$ objects from a set of $m$ objects, under the restriction that no objects can be chosen more than $q-1$ times, is equal to 
$$ 
N(t, m, q-1)=\sum_{j=0}^{m} (-1)^j \binom{m}{j} \binom{t-jq+m-1}{t-jq}.  
$$
The desired dimension then follows from (\ref{eqn-may241}). 
\end{proof} 

The formula in (\ref{eqn-mycode5221}) looks complicated. The following theorem documents a upper bound on the dimension of  the code $\C(q, m, r, \ell)$. 

\begin{theorem}\label{thm-dimmycode}
Let $\ell = r \ell_1 + r-1$, where $0 \leq \ell_1 \leq (\frac{q-1}r)m-2$. Let $\ell_2=\lceil\frac{\ell+1}{q-1}\rceil$. Then 
\begin{eqnarray*}
\dim(\C(q, m, r, \ell)) \leq \frac{q^m-\sum_{t=0}^{m-\ell_2} \binom{m}{t} (q-1)^t }{r}.  
\end{eqnarray*}  
\end{theorem}

\begin{proof}
Recall
$$
D_{(q,m, r, \ell)}=\{i \in \Z_N: 1 \leq \wt_q(i) < (q-1)m-\ell, \ i \equiv 1 \pmod{r}\}.  
$$ 
Note that $(q-1)m-\ell=(q-1)m-r\ell_1-r+1$. Then by definition, 
\begin{align*}
D_{(q,m, r, \ell)} &= \{i \in \Z_N: 1 \leq \wt_q(i) \leq (q-1)m-r\ell_1-2r+1, \ i \equiv 1 \pmod{r}\} \\
&\supseteq   \{i \in \Z_N: 1 \leq \wt_q(i) \leq (q-1)(m-\ell_2), \ i \equiv 1 \pmod{r}\} \\ 
&\supseteq   \{i \in \Z_N: 1 \leq \wt(i) \leq m-\ell_2, \ i \equiv 1 \pmod{r}\}.  
\end{align*} 
It then follows from Lemma \ref{lem-april91} that 
\begin{align*}
|D_{(q,m, r, \ell)}| & \geq   |\{i \in \Z_N: 1 \leq \wt(i) \leq m-\ell_2, \ i \equiv 1 \pmod{r}\}| \\
&= \sum_{t=1}^{m-\ell_2} \binom{m}{t} |\{(x_1,x_2, \ldots, x_t) \in \{1, 2, \cdots, q-1 \}^t: x_1+x_2+\cdots+x_t \equiv 1 \pmod{r} \}| \\ 
&= \sum_{t=1}^{m-\ell_2} \binom{m}{t}\frac{(q-1)^t}{r}.  
\end{align*} 
Consequently, 
\begin{eqnarray*}
\dim(\C(q, m, r, \ell)) = \frac{q^m-1}{r}-|D_{(q,m, r, \ell)}| \leq \frac{q^m-\sum_{t=0}^{m-\ell_2} \binom{m}{t} (q-1)^t }{r}. 
\end{eqnarray*}
The desired conclusion follows.
\end{proof} 

In order to determine the minimum distance of the $\lambda$-constacyclic code $\C(q, m, r, \ell)$, we will give another form 
of this code.

Let $\beta$ be a primitive element of $\gf(q^m)$ and let $\m_{\beta}(x)=\sum_{i=0}^{m-1}\epsilon_i x^i+x^m$, where $\epsilon_i\in \gf(q)$. Since $\m_{\beta}(x)$ is the minimal polynomial of $\beta$ over $\gf(q)$, $\epsilon_0 \neq 0$. 
Let
 $$\mathbf{M}=\begin{pmatrix}
0& 1& 0& \cdots &0\\
0& 0& 1& \cdots &0\\
\vdots &\vdots & \vdots &\ddots&\vdots \\
0& 0&0&\cdots& 1\\	
-\epsilon_0& -\epsilon_1 &-\epsilon_2&\cdots& -\epsilon_{m-1}	
\end{pmatrix}
 $$
be the companion matrix of $\m_{\beta}(x)$. Note that $n=\frac{q^m-1}r$ and $\beta^n=\lambda$, then $\mathbf{M}^n=\lambda \mathbf{E}$, where $\mathbf{E}$ is the identity matrix of order $m$. Furthermore, 
\begin{equation}\label{EEQ-6}
	\gf(q)^m=\{\0\}\cup\{\mathbf{e} \mathbf{M}^i:0\leq i\leq q^m-2 \}, 
\end{equation}
where $\0=(0,0,\ldots, 0)$ and $\mathbf{e}=(1,0,\ldots,0)$. It is clear that $\{1, \beta, \cdots, \beta^{m-1} \}$ is a basis for $\gf(q^m)$ as a vector space over $\gf(q)$. Let $\overline{\beta}=(1,\beta, \ldots, \beta^{m-1})$, then  
$$\mathbf{M} \overline{\beta}^T=( \beta, \beta^2,\ldots, \beta^{m-1}, -\sum_{j=0}^{m-1}\epsilon_j \beta^{j})^T=( \beta, \beta^2,\ldots, \beta^{m-1}, \beta^{m})^T=\beta \cdot \overline{\beta}^T.
 $$
 It follows that $\mathbf{M}^i \overline{\beta}^T= \beta^i \cdot \overline{\beta}^T$ for $0\leq i\leq q^m-2$. Therefore, $\beta^i=(\mathbf{e} \mathbf{M}^i, \overline{\beta})$ for $0\leq i\leq q^m-2$, where $(\cdot,\cdot)$ denotes the inner product of two vectors. It follows that the mapping $\0\mapsto 0$, $\mathbf{e} \mathbf{M}^i \mapsto \beta^i=(\mathbf{e} \mathbf{M}^i, \overline{\beta})$ is an isomorphism between the vector space structures of $\gf(q)^m$ and $\gf(q^m)$. Let $r-1\leq \ell<(q-1)m-1$ with $\ell \equiv r-1\pmod{r}$, and let $M(q, m, r, \ell)$ be the linear subspace of $\gf(q)[x_1,x_2, \dots, x_{m}]$, which is spanned by all monomials $x_1^{i_1}x_2^{i_2} \cdots x_{m}^{i_{m}}$ satisfying the following three conditions: 
\begin{enumerate}
\item $0\leq i_j\leq q-1$, for $1\leq j\leq m$,
\item $\sum_{j=1}^{m} i_j \equiv r-1 \pmod{r}$,
\item $\sum_{j=1}^{m} i_j  \leq \ell $.
\end{enumerate} 
Define
$${\mathcal{GC}}(q, m, r, \ell) =\left\{\bc_f=(f(\mathbf{e}),f(\mathbf{e} \mathbf{M}),\ldots, f(\mathbf{e} \mathbf{M}^{n-1})): f(x_1,x_2,\ldots, x_m)\in M(q, m, r, \ell) \right\}.$$
Below we will prove that ${\mathcal{GC}}(q, m, r, \ell)=\C(q, m, r, \ell)$. For this purpose, we need the following two lemmas.

\begin{lemma}\cite{DGM70}\label{slem-2}
	Let $f(x_1,x_2,\ldots, x_m)\in \gf(q)[x_1,x_2,\ldots, x_m]$, then we have the following:
\begin{enumerate}
    \item If $f(\mathbf{P})=0$ for all $\mathbf{P}\in \gf(q)^m$, then $f\equiv 0$.
	\item If $\deg(f)<(q-1)m$, then $\sum_{\mathbf{P}\in \gf(q)^m}f(\mathbf{P})=0$.
\end{enumerate}
\end{lemma}

\begin{lemma}\label{slem-1}
	Let $f(x_1,x_2,\ldots, x_m)\in M(q, m, r, \ell)$, then the following hold:
	\begin{enumerate}
	\item Let $0\leq i\leq n-1$ and $0\leq j\leq r-1$, then $f(\mathbf{e}\mathbf{M}^{jn+i})=\lambda^{-j}\cdot f(\mathbf{e} \mathbf{M}^i)$.
	\item If $f(\mathbf{e}\mathbf{M}^i)=0$ for all $0\leq i\leq n-1$, then $f\equiv 0$.  	
	\end{enumerate}
\end{lemma}

\begin{proof}
	1. Suppose $f(x_1,x_2,\ldots, x_m)=\sum c_{i_1,i_2,\cdots,i_m} x_1^{i_1}x_2^{i_2}\cdots x_{m}^{i_m}$. It follows from $\mathbf{M}^n=\lambda \mathbf{E}$ that $$\mathbf{e}\mathbf{M}^{j n+i}=\lambda^j \mathbf{e}\mathbf{M}^{i}.$$
	Note that $\sum_{k=1}^{m}i_k\equiv r-1\pmod {r}$ and $\ord(\lambda)=r$, then $(\lambda^j)^{i_1+i_2+\cdots+i_m}=\lambda^{-j}$. Consequently,
	$$f(\mathbf{e}\mathbf{M}^{j n+i})=f( \lambda^j \mathbf{e}\mathbf{M}^i)=\lambda^{-j} \cdot f(\mathbf{e}\mathbf{M}^i).$$
	
	2. Let $\mathbf{P}\in \gf(q)^m \backslash \{ \0 \}$, it follows from (\ref{EEQ-6}) that there are $0\leq i\leq n-1$ and $0\leq j\leq r-1$ such that $\mathbf{P}=\mathbf{e}\mathbf{M}^{j n+i}$. Then, $f(\mathbf{P})=\lambda^{-j} \cdot f(\mathbf{e}\mathbf{M}^{i})=0$. Note that $f(\mathbf{0})=0$. Therefore, $f(\mathbf{P})=0$ for all $\mathbf{P}\in \gf(q)^m$. By the first conclusion of Lemma \ref{slem-2}, we have $f\equiv 0$. The desired conclusion follows.
\end{proof}

\begin{theorem}\label{thm:32}
	Let notation be the same as before. Then ${\mathcal{GC}}(q, m, r, \ell)=\C(q, m, r, \ell)$.
\end{theorem}

\begin{proof}
First, we prove that ${\mathcal{GC}}(q, m, r, \ell)$ is a $\lambda$-constacyclic code of length $n$ over $\gf(q)$. For any $f(\mathbf{x})\in M(q, m, r, \ell)$, let $g(\mathbf{x})= \lambda\cdot f(\mathbf{x}\mathbf{M}^{n-1})$. It is easily verified that $g(\mathbf{x})\in M(q, m, r, \ell)$. By Lemma \ref{slem-1}, $$g(\mathbf{e}\mathbf{M}^i)=\lambda \cdot f(\mathbf{e}\mathbf{M}^{n+i-1})=\lambda \cdot \lambda^{-1}\cdot f(\mathbf{e} \mathbf{M}^{i-1})= f(\mathbf{e}\mathbf{M}^{i-1})$$
	for $1\leq i\leq n-1$. Therefore,
	$$\bc_g=(\lambda \cdot f(\mathbf{e}\mathbf{M}^{n-1}), f(\mathbf{e}),\ldots, f(\mathbf{e}\mathbf{M}^{n-2})) \in {\mathcal{GC}}(q, m, r, \ell).$$
	It follows that ${\mathcal{GC}}(q, m, r, \ell)$ is a $\lambda$-constacyclic code of length $n$ over $\gf(q)$. 
	
	Secondly, we prove that $\dim({\mathcal{GC}}(q, m, r, \ell))=\dim(\C(q, m, r, \ell))$. By the second conclusion of Lemma \ref{slem-1}, the evaluations of all monomials 
	$$\left \{x_1^{i_1}x_2^{i_2}\cdots x_m^{i_m}: \sum_{k=1}^{m}i_k\equiv r-1\pmod{r},\ 0\leq i_k\leq q-1, \ \sum_{k=1}^{m}i_k\leq \ell \right \}$$
 give linearly independent codewords. It follows that 
	\begin{align*}
		\dim({\mathcal{GC}}(q, m, r, \ell))&=|\{(i_1,i_2,\ldots,i_m)\in \{0,1,\cdots,q-1\}^m: \sum_{k=1}^{m}i_k\equiv r-1\pmod{r} \ \sum_{k=1}^{m}i_k\leq \ell \}|\\
		&=|\{ i\in \Z_N: \wt_q(i)\equiv r-1\pmod{r}, \ \wt_q(i) \leq \ell \}|.
	\end{align*}
	The desired dimension then follows from Theorem \ref{thm-mycode5221}.  
	
	Finally, we prove that ${\mathcal{GC}}(q, m, r, \ell)\subseteq \C(q, m, r, \ell)$. Let $g(x)$ be the generator polynomial of ${\mathcal{GC}}(q, m, r, \ell)$. Then we only need to prove $g_{(q, m, r, \ell)}(x)\mid g(x)$. Suppose that $i=\sum_{l=0}^{m-1}i_l q^l$, $\wt_q(i)<(q-1)m-\ell$ and $\wt_q(i)\equiv 1\pmod {r}$. For any $f\in M(q, m, r, \ell)$, let
	$$ c_f(x)=\sum_{j=0}^{n-1}f(\mathbf{e}\mathbf{M}^{j})x^j.$$
be the polynomial corresponding to the codeword $\bc_f\in {\mathcal{GC}}(q, m, r, \ell)$. For each $1\leq t\leq r-1$, by Lemma \ref{slem-1}, 
	\begin{align*}
		\sum_{j=0}^{n-1}f(\mathbf{e}\mathbf{M}^{t n+j})(\beta^i)^{t n+j}&=\sum_{j=0}^{n-1}f(\mathbf{e}\mathbf{M}^{j}) \lambda^{-t} (\beta^i)^{tn+j}\\
		&=\sum_{j=0}^{n-1}f(\mathbf{e}\mathbf{M}^{j}) (\beta^i)^{j}\\
		&=c_f(\beta^i).
	\end{align*}
Then we have
\begin{align*}
	c_f(\beta^i)&=\sum_{j=0}^{n-1}f(\mathbf{e}\mathbf{M}^j)(\beta^i)^j\\
	&=\frac{1}{r} \sum_{j=0}^{rn-1}f(\mathbf{e}\mathbf{M}^j)(\beta^i)^j\\
	&=\frac{1}{r} \sum_{j=0}^{q^m-2}f(\mathbf{e}\mathbf{M}^j)(\beta^j)^i\\
	&=\frac{1}{r} \sum_{j=0}^{q^m-2}f(\mathbf{e}\mathbf{M}^j)[(\mathbf{e}\mathbf{M}^j,\overline{\beta})]^i\\
	&=\frac{1}{r} \sum_{j=0}^{q^m-2}f(\mathbf{e}\mathbf{M}^j)[(\mathbf{e}\mathbf{M}^j,\overline{\beta})]^{\sum_{l=0}^{m-1}i_l q^l}\\
	&=\frac{1}{r} \sum_{j=0}^{q^m-2}f(\mathbf{e}\mathbf{M}^j)\prod_{l=0}^{m-1} [(\mathbf{e}\mathbf{M}^j,\overline{\beta})^{q^l}]^{i_l}.
\end{align*} 
For $\overline{e}M^j=(x_{1},x_{2},\ldots,x_{m})\in \gf(q)^m$,
\begin{align*}
h(\mathbf{e}\mathbf{M}^j):&=\prod_{l=0}^{m-1} [(\mathbf{e}\mathbf{M}^j,\overline{\beta})^{q^l}]^{i_l }\\
&=\prod_{l=0}^{m-1} [x_{1}+x_{2}\beta^{q^l}+\cdots+ x_{m} \beta^{(m-1)q^l}]^{i_l }	
\end{align*}
is a homogenous polynomial of degree $\wt_q(i)$ in indeterminates $x_j$. It is clear that 
$$\deg(fh)<\ell+(q-1)m-\ell=(q-1)m$$ 
and $f(\mathbf{0})h(\mathbf{0})=0$. It follows from Lemma \ref{slem-2} that 
 $$\sum_{j=0}^{q^m-2}f(\mathbf{e}\mathbf{M}^j)\prod_{l=0}^{m-1} [(\mathbf{e}\mathbf{M}^j,\overline{\beta})^{q^l}]^{i_l }=\sum_{\mathbf{P}\in \gf(q)^m} f(\mathbf{P})\prod_{l=0}^{m-1} [(\mathbf{P},\overline{\beta})^{q^l}]^{i_l }=0.$$
Therefore, $\beta^i$ is a root of $g(x)$. It follows that $g_{(q, m, r, \ell)}(x)$ divides $g(x)$. This completes the proof.
\end{proof}

It is hard to settle the minimum distance of the constacyclic code $\C(q, m, r, \ell)$. Below we develop some bounds on $d(\C(q, m, r, \ell))$.

\begin{theorem}\label{thm-april17}
Let $\ell= (q-1) \ell_1+\ell_0<(q-1)m-1$, where $0 \leq \ell_0\leq q-2$ and $\ell_0\equiv r-1\pmod{r}$. Then 
\begin{eqnarray}\label{eqn-LB11}
 \frac{(q-\ell_0)q^{m-\ell_1-1}-2}r+1 \leq d(\C(q,m, r, \ell)) \leq \frac{(q-\ell_0+r-2)q^{m-\ell_1-1}}{r}.
\end{eqnarray}  
\end{theorem}

\begin{proof} By definition, we have 
$$ 
(q-1)m-\ell = (q-1)(m-\ell_1-1)+q-1-\ell_0. 
$$
Let $H$ be the smallest integer with $\wt_q(H)=(q-1)m-\ell$. Then 
$$ 
H=(q-1-\ell_0)q^{m-\ell_1-1} + \sum_{i=0}^{m-\ell_1-2} (q-1)q^i = (q-\ell_0)q^{m-\ell_1-1}-1.  
$$
It is easily verified that every integer $u$ with $0 < u <H$ satisfies $\wt_q(u) 
<(q-1)m-\ell.$ Define 
$$ 
B=\left\{ 1+rj: 0 \leq j \leq \frac{(q-\ell_0)q^{m-\ell_1-1}-2}r-1  \right\}.  
$$
Then $B$ is a subset of $\{1,2,\cdots, H-r\}$ and $\beta^i$ is a zero of $\C(q,m, r, \ell)$ for each $i \in B$. The desired lower bound then follows from Lemma \ref{lem-BCHbound}. 

Let $${\mathcal{D}}(q, m, r, \ell) =\left \{ \widetilde{\mathbf{c}}_f=(f(\mathbf{e}\mathbf{M}^i))_{i=0}^{rn-1}: f(x_1,x_2,\ldots, x_{m})\in M(q, m, r, \ell) \right \}.$$	
 For any $f(x_1,x_2,\ldots, x_{m})\in M(q, m, r, \ell)$, it follows from Lemma \ref{slem-1} that 
 $$ \widetilde{\mathbf{c}}_f=(\bc_f \parallel  \lambda^{-1}\bc_f \parallel\cdots \parallel \lambda^{-(r-1)}\bc_f ), $$
 where $\mathbf{c}_f=(f(\mathbf{e}\mathbf{M}^i))_{i=0}^{n-1}$ and $\parallel$ denotes the concatenation of vectors. By Theorem \ref{thm:32}, we have 
 $$d(\C(q,m, r, \ell))=\frac{1}{r} \cdot d({\mathcal{D}}(q, m, r, \ell)).$$ 
 Let 
 $$f(x_1,x_2,\ldots, x_{m})=\prod_{i=1}^{\ell_1} [1-x_i^{(q-1)}]\cdot  x_{\ell_1+1}^{r-1} \cdot \prod_{i=1}^{\frac{\ell_0-r+1}r} [x_{\ell_1+1}^r-\omega^{r i}],$$
 where $\omega$ is a primitive element of $\gf(q)$. It is easily verified that
 $\deg(f)=(q-1)\ell_1+\ell_0$ and $f\in M(q, m, r, \ell)$. Clearly, $f(x_1,x_2,\ldots,x_m)$ is zero in $\gf(q)^m$ unless
 \begin{align}\label{eqn-wxq84}
& x_i =0,\ \mathrm{for}\ i=1, 2, \cdots, \ell_1, \nonumber \\
&  x_{\ell_1+1} \notin \{0\}\cup \left\{\lambda^j \omega^i: 1\leq i\leq  \frac{\ell_0-r+1}r, \ 0\leq j\leq r-1 \right \}	.
 \end{align}
 For any $1\leq i, i' \leq  \frac{\ell_0-r+1}r$, $\ 0\leq j, j'\leq r-1$, if $\lambda^{j'}\omega^{i'}=\lambda^{j}\omega^{i}$, then $\omega^{i-i'}=\lambda^{j'-j}$. It follows that $\omega^{r(i-i')}=1$. Then we have $\frac{q-1}r$ divides $i-i'$. Note hat $0\leq |i-i'|\leq \frac{\ell_0-2r+1}r<\frac{q-1}r$. Therefore, $i=i'$. Consequently, $j=j'$. That is to say, there are $[q-r( \frac{\ell_0-r+1}r )-1 ] q^{m-1-\ell_1}$ vectors in $\gf(q)^m\backslash \{\mathbf{0} \}$ satisfying both equations in (\ref{eqn-wxq84}) and $\wt(\bc_f)=(q-\ell_0+r-2)q^{m-1-\ell_1}$. It follows that $$d(\mathcal{D}(q, m, r, \ell))\leq (q-\ell_0+r-2)q^{m-1-\ell_1}.$$
 The desired upper bound follows.
\end{proof} 

If $r=2$ or $\ell_1=m-1$, it is easily verified that the upper and lower bounds in (\ref{eqn-LB11}) are equal. Therefore, we have following two conclusions.

\begin{corollary}\label{THM:24}
Let $\ell=(q-1)(m-1)+\ell_0$, where $m\geq 2$, $0\leq \ell_0 \leq q-2$ and $\ell_0\equiv r-1\pmod{r}$. Then 	
$$d(\C(q,m, r , \ell)) = \frac{q-\ell_0+r-2}r. $$
\end{corollary}

\begin{corollary}\label{thm:35}
Let $q$ be an odd prime power and $r=2$. Let $r-1\leq \ell=(q-1)\ell_1+\ell_0<(q-1)m-1$, where $m\geq 2$, $0\leq \ell_0\leq q-2$ and $\ell_0\equiv r-1 \pmod{ r }$. Then 	
$$d(\C(q,m, r, \ell)) = \left(\frac{q-\ell_0}2 \right)q^{m-1-\ell_1}.$$
\end{corollary}

\begin{example}
Let $(q, m, r, \ell)=(3, 3, 2, 3)$. Let $\beta$ be the primitive element of $\gf(3^3)$ with $\beta^3+2\beta+1=0$. Then $\C(3, 3, 2, 3)$ has parameters $[13, 10, 3]$ and is distance-optimal. 	
\end{example}

\begin{example}
Let $(q, m, r, \ell)=(5, 2, 2, 3)$. Let $\beta$ be the primitive element of $\gf(5^2)$ with $\beta^2+4\beta+2=0$. Then $\C(5, 2, 2, 3)$ has parameters $[12, 6, 5]$. Furthermore, $\C(5, 2, 2, 3)$ is self-dual and almost-distance optimal. 	
\end{example}

For $\ell$ with $r-1\leq \ell \leq  (q-1)(m-1)-1$ and $\ell \equiv r-1\pmod{r}$, we will give an improved bound on the minimum distance of the code $\C(q, m, r, \ell)$. To this end, we consider the subcode of the $\lambda$-constacyclic code $\C(q, m, r, \ell)$. Let $\widetilde{M}(q, m, r, \ell)$ be the linear subspace of $\gf(q)[x_1,x_2, \dots, x_{m}]$, which is spanned by all monomials $x_1^{i_1}x_2^{i_2} \cdots x_{m}^{i_{m}}$ satisfying the following three conditions: 
\begin{enumerate}
\item $0\leq i_j\leq q-1$, for $1\leq j\leq m$,
\item $\sum_{j=1}^{m} i_j \equiv \ell \pmod{q-1}$,
\item $\sum_{j=1}^{m} i_j  \leq \ell $.
\end{enumerate} 
It is easily verified that $\widetilde{M}(q, m, r, \ell)\subseteq M(q, m, r, \ell)$. In the special case $r=q-1$, $\widetilde{M}(q, m, r, \ell)=M(q, m, r, \ell)$. Associated with the $\lambda$-constacyclic code $\C(q, m, r, \ell)$ are the following two codes over $\gf(q)$:
$$\widetilde{\C}(q, m, r, \ell) =\left\{\widetilde{\bc}_f=(f(\mathbf{e}),f(\mathbf{e} \mathbf{M}),\ldots, f(\mathbf{e} \mathbf{M}^{n-1})): f(x_1,x_2,\ldots, x_m)\in \widetilde{M}(q, m, r, \ell) \right\},$$
and
$$\mathrm{P}(\widetilde{\C}(q, m, r, \ell)) =\left\{  \overline{\bc}_f=(f(\mathbf{e}),f(\mathbf{e} \mathbf{M}),\ldots, f(\mathbf{e} \mathbf{M}^{\overline{n}-1})): f(x_1,x_2,\ldots, x_m)\in \widetilde{M}(q, m, r, \ell) \right\},$$
where $\overline{n}=\frac{q^m-1}{q-1}$. In the special case $r=q-1$, $\C(q, m, r, \ell)$, $\widetilde{\C}(q, m, r, \ell)$ and $\mathrm{P}(\widetilde{\C}(q, m, r, \ell))$ are identical.

\begin{theorem}\label{THM:23}
Let notation be the same as before. Then the following hold:
\begin{enumerate}
\item The linear code $\widetilde{\C}(q, m, r, \ell)$ is a $\lambda$-constacyclic code of length $\frac{q^m-1}{r}$ over $\gf(q)$.
\item The $\lambda$-constacyclic code $\widetilde{\C}(q, m, r, \ell)\subseteq \C(q, m, r, \ell)$. In the special case $r=q-1$, 
$$\widetilde{\C}(q, m, r, \ell)= \C(q, m, r, \ell).$$
\item If $r=q-1$, the $\lambda$-constacyclic code $\widetilde{\C}(q, m, r, \ell)= \mathrm{P}(\widetilde{\C}(q, m, r, \ell))$. If $r<q-1$, the $\lambda$-constacyclic code $$\widetilde{\C}(q, m, r, \ell)=\{ (\overline{\bc}_f\| \omega^{\ell}\cdot \overline{\bc}_f\|\cdots \|\omega^{(\frac{q-1}{r}-1)\ell}\cdot \overline{\bc}_f):\overline{\bc}_f \in \mathrm{P}(\widetilde{\C}(q, m, r, \ell) )\},$$
where $\omega=\beta^{\overline{n}}$ is a primitive element of $\gf(q)$.
\end{enumerate}	
\end{theorem}

\begin{proof}
	1. The proof is similar to Theorem \ref{thm:32}, and the details are omitted here. 
	
	2. Note that $\widetilde{M}(q, m, r, \ell)\subseteq M(q, m, r, \ell)$ and  $\widetilde{M}(q, m, r, \ell)= M(q, m, r, \ell)$ for $r=q-1$. The desired result follows.
	
	3. If $r=q-1$, the desired result is obvious. If $r<q-1$, since $\beta$ is a primitive element of $\gf(q^m)$, $\omega:=\beta^{\overline{n}}$ is a primitive element of $\gf(q)$. Therefore, $\mathbf{M}^{\overline{n}}=\omega \mathbf{E}$. It follows that $$\mathbf{e}\mathbf{M}^{j  \overline{n}+i}=\omega^j \mathbf{e}\mathbf{M}^{i}$$ 
	for any $0\leq j\leq \frac{q-1}r-1$, $0\leq i\leq \overline{n}-1$. Let $f(x_1,x_2,\ldots, x_m)\in \widetilde{M}(q, m, r, \ell)$, then $$f(\mathbf{e}\mathbf{M}^{j\overline{n}+i})=f(\omega^{j} \mathbf{e}\mathbf{M}^{i})=\omega^{j\ell}\cdot f(\mathbf{e}\mathbf{M}^{i}).$$
	 Let $$\widetilde{\bc}_f=(f(\mathbf{e}),f(\mathbf{e} \mathbf{M}),\ldots, f(\mathbf{e} \mathbf{M}^{n-1}))$$ be the codeword corresponding to the polynomial $f(x_1,x_2,\ldots,x_m)$. Then
		$$\widetilde{\bc}_f=(\overline{\bc}_f\| \omega^{\ell}\cdot \overline{\bc}_f\|\cdots \|\omega^{(\frac{q-1}{r}-1)\ell}\cdot \overline{\bc}_f).$$
		The third desired result follows.
\end{proof}

Below we will prove that the code $\mathrm{P}(\widetilde{\C}(q, m, r, \ell))$ is scalar-equivalent to the projective Reed-Muller code $\mathrm{PRM}(q, m, \ell)$.

Let $T$ be the mapping from $\gf(q)[x_1,x_2,\ldots,x_m]$ to the quotient ring $$\gf(q)[x_1,x_2,\ldots,x_m]/\langle x_1^q-x_1,x_2^q-x_2,\ldots, x_m^q-x_m \rangle$$ defined by
\begin{align*}
T: \	\gf(q)[x_1,x_2,\ldots,x_m] &\rightarrow \gf(q)[x_1,x_2,\ldots,x_m]/\langle x_1^q-x_1,x_2^q-x_2,\ldots, x_m^q-x_m \rangle\\
f=\sum c_{i_1,i_2,\ldots,i_m}x_1^{i_1} x_2^{i_2}\cdots x_{m}^{i_m} &\mapsto \sum c_{i_1,i_2,\ldots,i_m}x_1^{i_1'} x_2^{i_2'}\cdots x_{m}^{i_m'} 
\end{align*}
where these $i_j'$ satisfy the following conditions:
\begin{enumerate}
\item If $i_j=0$, then $i_j'=0$.
\item If $i_j>0$, then $1\leq i_j'\leq q-1$ and $i_j'\equiv i_j\pmod{q-1}$.	
\end{enumerate}
For each $a\in \gf(q)$, $a^q=a$. It follows that $a^{i_j}=a^{i_j'}$. Consequently, $f(\mathbf{x})=T(f)(\mathbf{x})$ for any $\mathbf{x}\in \gf(q)^m$. 

Recall that $A(q, m, \ell)$ denotes the subspace of $\gf(q)[x_1, x_2,\dots, x_m]$ generated by all the homogeneous polynomials of degree $\ell$, where $\ell<(q-1)m$. 

 \begin{lemma}\label{lem:37}
 Let notation be the same as before. Let $\overline{n}=\frac{q^m-1}{q-1}$, where $m\geq 2$. Then $$\{ \mathbf{e} \mathbf{M}^i:0\leq i\leq \overline{n}-1 \}$$ is the set of points in $\PG(m-1,\gf(q))$.
 \end{lemma}

\begin{proof}
Suppose there are $0\leq i<j\leq \overline{n}-1$ such that $\mathbf{e}\mathbf{M}^i=\gamma \cdot  \mathbf{e} \mathbf{M}^j$, where $\gamma\in \gf(q)^*$. Then
\begin{align*}
	\beta^i&=(\mathbf{e}\mathbf{M}^i, \overline{\beta})\\
	&= (\gamma \cdot\mathbf{e}\mathbf{M}^j, \overline{\beta})\\
	&=\gamma \cdot(\mathbf{e}\mathbf{M}^j, \overline{\beta})\\
	&=\gamma \cdot \beta^j.
\end{align*}
It follows that $\beta^{j-i}\in \gf(q)^*$, which deduces that $\overline{n}\mid (j-i)$, a contradiction. Therefore, any two distinct elements in the set $\{ \mathbf{e} \mathbf{M}^i:0\leq i\leq \overline{n}-1 \}$ are linearly independent over $\gf(q)$. The desired result follows.
\end{proof}

According to Lemma \ref{lem:37}, the projective Reed-Muller code $\PRM(q, m, \ell)$ is scalar-equivalent to the code
$$\widehat{\C} (q, m, \ell)=\left\{\widehat{\bc}_f=(f(\mathbf{e}),f(\mathbf{e} \mathbf{M}),\ldots, f(\mathbf{e} \mathbf{M}^{\overline{n}-1})): f(x_1,x_2,\ldots, x_m)\in A(q, m, \ell ) \right\}, $$
where $\overline{n}=\frac{q^m-1}{q-1}$.

\begin{theorem}\label{THM:30}
Let notation be the same as before. Let $r-1\leq \ell \leq (q-1)(m-1)-1$ and $\ell \equiv r-1\pmod{r}$. Then we have the following hold.
\begin{enumerate}
\item The code $\mathrm{P}(\widetilde{\C}(q, m, r, \ell))$ has dimension 
$$\sum_{t \equiv \ell \pmod{q-1} \atop 0 < t \leq \ell} \left(  \sum_{j=0}^m (-1)^j \binom{m}{j} \binom{t-jq+m-1}{t-jq} \right).$$
\item The projective Reed-Muller code $\PRM(q, m, \ell)$ is scalar-equivalent to the code $$\mathrm{P}(\widetilde{\C}(q, m, r, \ell)).$$
In particular, if $r=q-1$, the code $\PRM(q, m, \ell)$ is scalar-equivalent to the $\lambda$-constacyclic code $\C(q, m, r, \ell)$.
\item Let $\ell= (q-1)\ell_1+\ell_0$, where $\ell_1\geq 0$ and $0<\ell_0<q-1$. then the code $\mathrm{P}(\widetilde{\C}(q, m, r, \ell))$ has minimum distance $$(q-\ell_0+1)q^{m-2-\ell_1}.$$ 
\item The $\lambda$-constacyclic code $\widetilde{\C}(q, m, r, \ell)$ has parameters $[n, k, d]$, where
\begin{align*}
n=&\frac{q^m-1}{r},\\	
k=&\sum_{t \equiv \ell \pmod{q-1} \atop 0 < t \leq \ell} \left(  \sum_{j=0}^m (-1)^j \binom{m}{j} \binom{t-jq+m-1}{t-jq} \right),\\
d=&\left(\frac{q-1}r\right)(q-\ell_0+1)q^{m-2-\ell_1}.
\end{align*}
\end{enumerate}
\end{theorem}

\begin{proof}
1. Similar to Theorem \ref{thm:32}, one can prove that 
\begin{align*}
		\dim(\mathrm{P}(\widetilde{\C}(q, m, r, \ell)))&=|\{(i_1,i_2,\ldots,i_m)\in \{0,1,\cdots,q-1 \}^m: \sum_{k=1}^{m}i_k\equiv \ell \pmod{q-1},  \sum_{k=1}^{m}i_k\leq \ell \}|.
	\end{align*}	
	The remaining proofs are similar to Theorem \ref{thm-mycode5221}, and details are omitted here.
	
	2. We claim that $\widehat{\C}(q, m, \ell)=\mathrm{P}(\widetilde{\C}(q, m, r, \ell))$. It follows from Theorem \ref{thm-PRMcode7} and Result 1 that $\dim(\widehat{\C}(q, m, \ell))=\dim(\mathrm{P}(\widetilde{\C}(q, m, r, \ell)))$. Therefore, in order to prove the desired conclusion, we only need to prove $\widehat{\C}(q, m, \ell)\subseteq \mathrm{P}(\widetilde{\C}(q, m, r, \ell))$. Let $f(x_1,x_2,\ldots,x_m)\in A(q, m, \ell)$ and let
	$$ \widehat{\bc}_f=(f(\mathbf{e}),f(\mathbf{e} \mathbf{M}),\ldots, f(\mathbf{e} \mathbf{M}^{\overline{n}-1}))$$
	be the codeword corresponding to the polynomial $f(x_1,x_2,\ldots,x_m)$. Since $f(\mathbf{x})=T(f)(\mathbf{x})$ for any $ \mathbf{x}\in \gf(q)^m$, we have $\widehat{\bc}_{f}=(T(f)(\mathbf{e}),T(f)(\mathbf{e} \mathbf{M}),\ldots, T(f)(\mathbf{e} \mathbf{M}^{n-1}))$.
	
	 If $f(x_1,x_2,\ldots,x_m)=\sum c_{i_1,i_2,\ldots,i_m} x_1^{i_1}x_2^{i_2}\cdots x_m^{i_m}\in A(q, m, \ell)$, then $T(f)=\sum c_{i_1,i_2,\ldots,i_m} x_1^{i_1'}x_2^{i_2'}\cdots x_m^{i_m'}$ where $i_j'=i_j=0$ or $1\leq i_j'\leq q-1$ and such that $i_j'\equiv i_j\pmod{q-1}$. It is clear that $$\sum_{j=1}^{m} i_j'\equiv \sum_{j=1}^{m} i_j\equiv \ell  \pmod{q-1} ,$$
	 and $\deg( T(f) )\leq \ell$. Therefore, $T(f)\in \widetilde{M}(q, m, r, \ell)$. It follows that $$\widehat{\bc}_{f}=(T(f)(\mathbf{e}),T(f)(\mathbf{e} \mathbf{M}),\ldots, T(f)(\mathbf{e} \mathbf{M}^{n-1}))\in \mathrm{P}(\widetilde{\C}(q, m, r, \ell)).$$
	Consequently, $\widehat{\C}(q, m, \ell)\subseteq \mathrm{P}(\widetilde{\C}(q, m, r, \ell))$. This proves the claim.
	
	Note that the code $\PRM(q, m, \ell)$ is scalar-equivalent to the code $\widehat{\C}(q, m, \ell)$, and $\mathrm{P}(\widetilde{\C}(q, m, r, \ell)) = \C(q, m, r, \ell)$ for $r=q-1$. The desired result follows.
	
	3. The desired result follows from Result 2 and Theorem \ref{thm-PRMcode7}. 
	
	4. By Result 3 of Theorem \ref{THM:23}, $\dim(\widetilde{\C}(q, m, r, \ell)=\dim(\mathrm{P}(\widetilde{\C}(q, m, r, \ell)))$ and  $$d(\widetilde{\C}(q, m, r, \ell)= \left(\frac{q-1}r\right)\cdot d(\mathrm{P}(\widetilde{\C}(q, m, r, \ell)))$$
	The desired result follows from Result 1 and Result 3.
\end{proof}

Note that $\widetilde{\C}(q, m, r, \ell)\subseteq \C(q, m, r, \ell)$, we have $d(\C(q, m, r, \ell))\leq d(\widetilde{\C}(q, m, r, \ell))$. By Result 4 of Theorem \ref{THM:30}, we can improve the upper bound in (\ref{eqn-LB11}).

\begin{theorem}
Let $r>2$ and $r\mid(q-1)$. Let $\ell= (q-1) \ell_1+\ell_0\leq (q-1)(m-1)-1$, where $\ell_1 \leq m-2$, $0 \leq \ell_0\leq q-2$ and $\ell_0\equiv r-1\pmod{r}$. Then 
 \begin{equation}\label{EEQ15}
 	\frac{(q-\ell_0)q^{m-\ell_1-1}-2}r+1 \leq d(\C(q,m, r, \ell)) \leq \left(\frac{q-1}r\right)(q-\ell_0+1)q^{m-2-\ell_1}.
 \end{equation}
\end{theorem}

When $r=q-1$, the minimum distance of the code $\C(q, m, r, \ell)$ just takes the upper bound in (\ref{EEQ15}).

\begin{corollary}\label{thm:38}
	Let $q\geq 3$ be a prime power and $r=q-1$. Let $\ell=(q-1)\ell_1+q-2$, where $0\leq \ell_1\leq m-2$. Then the code $\C(q, m, q-1, \ell)$ is scalar-equivalent to $\PRM(q, m, \ell)$. Furthermore, $$d(\C(q, m, q-1, \ell))=3\cdot q^{m-2-\ell_1}.$$
\end{corollary}

\begin{proof}
	 When $r=q-1$, $\C(q, m, r, \ell)=\mathrm{P}(\widetilde{\C}(q, m, r, \ell))$. The desired result follows directly from Theorem \ref{THM:30}.
\end{proof}

Corollary \ref{thm:38} shows that the constacyclic code $\C(q, m, q-1,(q-1)\ell_1+q-2)$ is scalar-equivalent to the projective Reed-Muller code $\PRM(q, m,(q-1)\ell_1+q-2)$, $0\leq \ell_1\leq m-2$. Although the two $q$-ary codes $\C(q, m, q-1, (q-1)\ell_1+q-2)$ and $\PRM(q, m,(q-1)\ell_1+q-2)$ are scalar-equivalent, the former is more interesting, as the former is a constacyclic code but the later is a linear code.

\begin{example}
	Let $(q, m, r, \ell)=(3, 4, 2, 1)$. Let $\beta$ be the primitive element of $\gf(3^4)$ with $\beta^4+2\beta^3+2=0$. Then $\C(3, 4, 2, 1)$ has parameters $[40, 4, 27]$ and is distance-optimal. 	
\end{example}

When $\ell_1=m-2$ and $\ell_0=r-1$, it is easily verified that the upper and lower bounds in (\ref{EEQ15}) are equal. Then we have the following conclusion.

\begin{corollary}\label{THM:34}
Let $\ell= (q-1)(m-2)+r-1$, where $m\geq 2$. Then 
\begin{eqnarray*}
 d(\C(q,m, r, \ell))=\frac{(q-1)(q-r+2)}{r}.
\end{eqnarray*}  
\end{corollary}

The following problem is interesting and worth of investigation.

\begin{open} 
Let $q\geq 7$ and $m\geq 2$. Let $r$ be a divisor of $q-1$ and $2<r<q-1$. Let $\ell=r \ell_1+r-1 $, where $ 1\leq \ell_1 \leq (\frac{q-1}{r})m-3$. Determine the minimum distance of the code $\C(q, m, r, \ell )$ or improve the lower bound in (\ref{eqn-LB11}).
\end{open} 

We have the following results about the dual code of the constacyclic code $\C(q, m, r, \ell)$.

\begin{theorem}\label{thm:40}
	Let $\ell=r\ell_1+r-1$, where $0\leq \ell_1\leq (\frac{q-1}r)m-2$. Then the dual code $\C(q, m, r, \ell)^{\bot}$ of the constacyclic code $\C(q, m, r, \ell)$ is the $\lambda^{-1}$-constacyclic code of length $\frac{q^m-1}r$ over $\gf(q)$ with generator polynomial 
	\begin{align*}
	g_{(q, m, r, \ell)}^{\bot}(x)=\prod_{i\in \Gamma_{(q,N,r)}^{(r-1)}  \atop \wt_q(i)\leq \ell }\m_{\beta^i}(x),
	\end{align*}
	where $\Gamma_{(q,N,r)}^{(r-1)}=\{ i\in \Gamma_{(q,N)}: \wt_q(i)\equiv r-1\pmod{r} \}$.	In particular, if $r=2$, then $$\C(q, m, r, \ell)^{\bot}=\C(q, m, r, (q-1)m-\ell-r ).$$ 
 \end{theorem}

\begin{proof}
It is clear that
\begin{align*}
	x^n-\lambda &=\prod_{i\in \Gamma_{(q,N,r)}^{(1)} }\m_{\beta^i}(x)\\
	&=\prod_{i\in \Gamma_{(q,N,r)}^{(1)}  \atop \wt_q(i)<(q-1)m-\ell }\m_{\beta^i}(x)  \prod_{i\in \Gamma_{(q,N,r)}^{(1)}  \atop \wt_q(i)\geq (q-1)m-\ell }\m_{\beta^i}(x).
	\end{align*}
It follows that the generator polynomial of $\C(q, m, r, \ell)^{\bot}$ is 
\begin{align*}
	g_{(q, m, r, \ell)}^{\bot}(x):&=\prod_{i\in \Gamma_{(q,N,r)}^{(1)}  \atop \wt_q(i)\geq (q-1)m-\ell } \widehat{\m_{\beta^i}}(x)\\
	&=\prod_{i\in \Gamma_{(q,N,r)}^{(1)}  \atop \wt_q(i)\geq (q-1)m-\ell }\m_{\beta^{N-i}}(x),
	\end{align*}
	where $\widehat{\m_{\beta^i}}(x)$ denotes the reciprocal polynomial of $\m_{\beta^i}(x)$. Since $r\mid N$, $i\equiv 1\pmod{r}$ if and only if $N-i\equiv r-1\pmod {r}$. Note that $$\wt_q(N-i)=(q-1)m-\wt_q(i),$$
	then $\wt_q(i)\geq (q-1)m-\ell$ if and only if $\wt_q(N-i)\leq \ell$. Therefore,
	\begin{align*}
	g_{(q, m, r, \ell)}^{\bot}(x)=\prod_{i\in \Gamma_{(q,N,r)}^{(r-1)}  \atop \wt_q(i)\leq \ell }\m_{\beta^i}(x).
	\end{align*}
    When $r=2$, it is clear that $g_{(q, m, r, \ell)}^{\bot}(x)=g_{(q, m, r, (q-1)m -\ell -r)}(x)$. This completes the proof.
\end{proof}

\begin{theorem}
Let $r>1$ and $r\mid (q-1)$. Let $\ell=r \ell_1+r-1$, where $0 \leq \ell_1 \leq (\frac{q-1}r)m-2 $. Then 
$$
\dim(\C(q,m, r, \ell)^\perp)=\frac{q^m-1}{r}-\sum_{t=0}^{\ell_1}\sum_{j=0}^m (-1)^j \binom{m}{j} \binom{tr+ r-1-jq+m-1}{tr+ r-1-jq}
$$ 
and 
\begin{eqnarray}
d(\C(q,m, r, \ell )^\perp) \geq \left\lfloor \left(\frac{\ell'_0+1}{r}\right)q^{\ell'_1} \right\rfloor+1,
\end{eqnarray}
where $\ell'_0, \ell'_1$ such that $\ell=(q-1)\ell'_1+\ell'_0$ and $0\leq \ell'_0\leq q-2$.
\end{theorem}

\begin{proof} 
The desired conclusion on the dimension of the dual code follows from Theorem \ref{thm-mycode5221}. Below we prove the lower bound on the minimum distance of the dual code.

Suppose $\ell=(q-1)\ell_1'+\ell_0'$, where $0\leq \ell_0'\leq q-2$. Let $H$ be the smallest integer with $\wt_q(H)=\ell$. Then 
$$ 
H=\ell_0'  q^{\ell_1'} + \sum_{i=0}^{\ell_1'-1} (q-1)q^i = (\ell_0'+1) q^{\ell_1'}-1.  
$$
It is easily verified that every integer $u$ with $0 < u \leq H$ satisfies $\wt_q(u)\leq \ell.$ Define 
$$ 
B=\left\{ r-1+r j: 0 \leq j \leq \left\lfloor \left(\frac{\ell_0'+1}{r}\right)q^{\ell_1'} \right\rfloor -1  \right\}.  
$$
Then $\beta^i$ is a zero of $g_{(q, m, r, \ell)}^{\bot}(x)$ for each $i \in B$. The desired bound then follows from Lemma \ref{lem-BCHbound}. This completes the proof. 
\end{proof}

When $r=2$ or $q-1$, the minimum distance of $\C(q, m, r,\ell)^{\bot}$ can be completely determined.

\begin{theorem}\label{thm:44}
Let $q$ be an odd prime power and $r=2$. Let $1\leq \ell=(q-1)\ell_1+\ell_0<(q-1)m-1$, where $m\geq 2$, $0\leq \ell_0\leq q-2$ and $\ell_0\equiv r-1\pmod {r}$. Then 
$$d(\C(q, m, r, \ell)^{\bot})=\begin{cases}
	\left(\frac{3+\ell_0}2\right) q^{\ell_1}~&{\rm if}~\ell_0<q-2,\\
	q^{\ell_1+1}~&{\rm if}~\ell_0=q-2.
\end{cases}
$$
\end{theorem}

\begin{proof}
By Theorem \ref{thm:40}, $\C(q, m, r, \ell)^{\bot}=\C( q, m, r, (q-1)m-\ell -r)$. Note that 
$$(q-1)m-\ell-r=(q-1)(m-\ell_1-1)+q-3-\ell_0.$$
If $\ell_0<q-2$, by Corollary \ref{thm:35}, $d(\C(q, m, r, \ell)^{\bot})=(\frac{3+\ell_0}2) q^{\ell_1}$. If $\ell_0=q-2$, by Corollary \ref{thm:35}, $$d(\C(q, m, r, \ell)^{\bot})=q^{\ell_1+1}.$$ 
This completes the proof.
\end{proof}

\begin{theorem}\label{thm:45}
Let $q>2$ be a prime power and $r=q-1$. Let $\ell=(q-1)\ell_1+q-2$, where $0\leq \ell_1 \leq m-2$. Then 
$d(\C(q, m, r, \ell)^{\bot})=q^{\ell_1+1}.$
\end{theorem}

\begin{proof}
	By Corollary \ref{thm:38}, the code $\C(q, m, q-1, \ell)$ is scalar-equivalent to $\PRM(q, m, \ell)$. It follows that the code $\C(q, m, q-1, \ell)^{\bot}$ is scalar-equivalent to $\PRM(q, m, \ell)^{\bot}$. It then follows from Theorem \ref{prm-8} that $\PRM(q, m, \ell)^{\bot}=\PRM(q, m, (m-1)(q-1)-\ell)$. Thereby, $$d(\C(q, m, r, \ell))=d( \PRM(q, m, (m-1)(q-1)-\ell)).$$ 
Note that $(m-1)(q-1)-\ell-1=(m-2-\ell_1)(q-1)$, by Theorem \ref{thm-PRMcode7}, $$d( \PRM(q, m, (m-1)(q-1)-\ell))=q^{\ell_1+1}.$$
This completes the proof. 
\end{proof}

\subsection{Some special cases of the constacyclic code $\C(q, m, r, \ell)$}

In this subsection, we will study the code $\C(q, m, r, \ell)$ in some special cases. We do have the dimension formula of the code $\C(q, m, r, \ell)$ in Theorem \ref{thm-mycode5221}, which may not be easily simplified. Instead, we will determine the generator or check polynomial of $\C(q, m, r, \ell)$ and will then know the dimension of the code without using Theorem \ref{thm-mycode5221}.

\begin{theorem}\label{thm:31}
Let $r>1$ and $r \mid (q-1)$. Let $\ell=(q-1)m-r-1$, where $m\geq 2$.  Then the constacyclic code $\C(q, m, r, \ell)$ has parameters $[\frac{q^m-1}{r}, \frac{q^m-1}r-m, d]$, where 
\begin{equation*}
	d=\begin{cases}
	2 & {\rm if~}r<q-1,\\
	3 &{\rm if~} r=q-1.
\end{cases}
\end{equation*}
Moreover, the dual code $\C(q, m, r, \ell)^{\bot}$ has parameters $[\frac{q^m-1}{r}, m, (\frac{q-1}r)q^{m-1}]$.
\end{theorem}

\begin{proof}
When $\ell=(q-1)m-r-1$, we have  
\begin{align*}
D_{(q, m, r, \ell)}=&\{i\in \Z_N: \wt_q(i)<r+1,\ \wt_q(i)\equiv 1\pmod{r} \}\\
=&\{i\in \Z_N: \wt_q(i)=1\}=C_1^{(q,N)}.	
\end{align*}
By definition, we have $g_{(q, m, r,\ell)}(x)=\m_{\beta}(x)$. Then $$\dim(\C(q, m, r,\ell))=\frac{q^m-1}r-m.$$

Now consider the minimum distance of the code $\C(q, m, r, \ell)$. There are two cases.
\begin{enumerate}
\item $r<q-1$. Then $\ell=(q-1)(m-1)+q-2-r$. By Theorem \ref{thm-april17}, $d(\C(q,m, r, \ell))= 2$.	
\item $r=q-1$. Then $\ell=(q-1)(m-2)+(q-2)$. By Theorem \ref{thm-april17}, $d(\C(q,m, r, \ell))\geq 3$. By the Sphere Packing bound, $d(\C(q, m, r, \ell ))\leq 3$. The desired result follows. 
\end{enumerate}

By Lemma \ref{lem-01}, the trace representation of $\C(q, m, r, \ell)^{\bot}$ is given by 
$$
\C(q,m, r,\ell)^{\bot}=\{\bc(a)=(\tr_{q^m/q}(a \beta^{i}))_{i=0}^{n-1}: a \in \gf(q^m) \}. 
$$
For $a\in \gf(q^m)^*$,
\begin{align*}
\wt(\bc(a))&=n-\frac{1}{q}\sum_{i=0}^{n-1}\sum_{x\in \gf(q)} \zeta_p^{\tr_{q/p}(x \tr_{q^m/q}(a \beta^i))}  	\\
&=n-\frac{1}{rq}\sum_{i=0}^{rn-1}\sum_{x\in \gf(q)}\zeta_p^{\tr_{q/p}(x \tr_{q^m/q}(a \beta^i))}\\
&=n-\frac{1}{rq} \sum_{x\in \gf(q)} \sum_{y\in \gf(q^m)^*} \zeta_p^{\tr_{q^m/p}(a x y)}\\
&=n-\frac{1}{rq} [q^m-1+(q-1)\sum_{y\in \gf(q^m)^*} \zeta_p^{\tr_{q^m/p}(y)}]\\
&=n-\frac{1}{rq} [q^m-1-(q-1)]\\
&=(\frac{q-1}r)q^{m-1}.
\end{align*}
	The desired minimum distance of the dual code then follows.
\end{proof}  

Notice that the constacyclic code $\C(q, m, q-1, (q-1)(m-2)+q-2)$ has parameters 
$$ 
\left[  \frac{q^m-1}{q-1},   \frac{q^m-1}{q-1}-m,  3 \right] 
$$ 
and is monomially-equivalent to the Hamming code. The dual code $\C(q, m, (q-1)(m-2)+q-2)^\perp$ has parameters 
$$ 
\left[  \frac{q^m-1}{q-1},  m,  q^{m-1} \right] 
$$ 
and is monomially-equivalent to the Simplex code.

\begin{theorem}
Let $r>1$ and $r\mid(q-1)$. Let $\ell=(q-1)(m-1)+\ell_0$, where $m\geq 2$, $0\leq \ell_0< q-2$ and $\ell_0\equiv r-1\pmod{r}$. Then the constacyclic code $\C(q,m, r , \ell)$ has parameters	
$$\left[\frac{q^m-1}{r}, \frac{q^m-1}{r}-\sum_{t=0}^{\frac{q-2-\ell_0}r-1}\binom{m+rt}{rt+1},\frac{q-\ell_0+r-2}r\right]. $$
\end{theorem}

\begin{proof}
When $\ell=(q-1)(m-1)+\ell_0$, we have $(q-1)m-\ell=q-1-\ell_0$. Then 
\begin{align*}
D_{(q, m, r, \ell)}=&\{i\in \Z_N: \wt_q(i)<q-1-\ell_0,\ \wt_q(i)\equiv 1\pmod{r} \}\\
=&\bigcup_{t=0}^{\frac{q-2-\ell_0}r-1}\{i\in \Z_N: \wt_q(i)= r t+1 \}.
\end{align*}
It is easily checked that
\begin{align*}
|D_{(q, m, r, \ell)}|=\sum_{t=0}^{\frac{q-2-\ell_0}r-1}\binom{m+rt}{rt+1}. 
\end{align*}
The desired dimension follows. The desired minimum distance then follows from Corollary \ref{THM:24}.
\end{proof}

\begin{theorem}\label{}
Let $r>1$ and $r\mid(q-1)$. Let $\ell=(q-1)(m-2)+r-1$, where $m\geq 2$. Then the constacyclic code $\C(q, m, r, \ell)$ has parameters $$\left [\frac{q^m-1}{r}, \frac{q^m-1}r-\kappa, \frac{(q-1)(q-r+2)}{r} \right],$$
where 
$$\kappa=\begin{cases}
\sum_{t=0}^{\frac{2(q-1-r)}r}\binom{m+rt}{rt+1} &{\rm if}~ \frac{q+1}2\leq r\leq q-1,\\
\sum_{t=0}^{\frac{2(q-1-r)}{r}}\binom{m+rt}{rt+1}-m \sum_{t=\frac{q-1}{r}}^{\frac{2(q-1-r)}{r}}\binom{tr-q+m}{tr-q+1} &{\rm if}~ 2 \leq r\leq \frac{q-1}2.
\end{cases}
 $$
\end{theorem}

\begin{proof}
When $\ell=(q-1)(m-2)+r-1$, we have $(q-1)m-\ell=2(q-1)-r+1$. Then
\begin{align*}
D_{(q, m, r, \ell)}=&\{i\in \Z_N: \wt_q(i)\leq 2(q-1)-2r+1,\ \wt_q(i)\equiv 1\pmod{r} \}\\
&=\bigcup_{t=0}^{\frac{2(q-1-r)}{r}}\{i\in \Z_N: \wt_q(i)=rt+1 \}. 
\end{align*}
When $\frac{q+1}2\leq r\leq q-1$, we have $2(q-1)-2r+1\leq q-2$. It is easily checked that
\begin{align*}
|D_{(q, m, r, \ell)}|=\sum_{t=0}^{\frac{2(q-1-r)}r}\binom{m+rt}{rt+1}.
\end{align*} 
When $2\leq r\leq \frac{q-1}{2}$, we have $q\leq 2(q-1)-2r+1<2(q-1)$. Clearly, 
\begin{align*}
|D_{(q, m, r, \ell)}|&=\sum_{t=0}^{\frac{q-1}{r}-1}|\{i\in \Z_N: \wt_q(i)=rt+1 \}|+ \sum_{t=\frac{q-1}{r}}^{\frac{2(q-1-r)}{r}}|\{i\in \Z_N: \wt_q(i)=rt+1 \}|\\
&=\sum_{t=0}^{\frac{q-1}{r}-1}\binom{m+rt}{rt+1}+\sum_{t=\frac{q-1}{r}}^{\frac{2(q-1-r)}{r}}\left[\binom{m+rt}{rt+1}-m \binom{tr-q+m}{tr-q+1}\right]\\
&=\sum_{t=0}^{\frac{2(q-1-r)}{r}}\binom{m+rt}{rt+1}-m \sum_{t=\frac{q-1}{r}}^{\frac{2(q-1-r)}{r}}\binom{tr-q+m}{tr-q+1}.
\end{align*}
The desired dimension follows. The desired minimum distance then follows from Corollary \ref{THM:34}.
\end{proof}

\begin{theorem}\label{thm-april181}
Let $q \geq 3$ and $m\geq 3$. Then the constacyclic code $\C(q, m, q-1, (q-1)(m-3)+q-2)$ has parameters 
$$ 
\left[  \frac{q^m-1}{q-1},   \frac{q^m-1}{q-1}-\binom{m+q-1}{q}, 3\cdot q   \right], 
$$
and the dual code $\C(q, m, q-1, (q-1)(m-3)+q-2)^{\bot}$ has parameters 
$$ 
\left[  \frac{q^m-1}{q-1},  \binom{m+q-1}{q}, q^{m-2}   \right]. 
$$
\end{theorem}

\begin{proof}
Let $\ell = (q-1)(m-3)+q-2$. Note that $(q-1)m-\ell=2(q-1)+1$. It is easy to see that 
\begin{align*}
D_{(q, m, q-1, \ell)}&=\{i\in \Z_{N}: \wt_q(i)<2(q-1)+1,\ \wt_q(i)\equiv 1\pmod{q-1} \}\\
&=\{i\in \Z_{N}: \wt_q(i)=1 \} \cup \{i\in \Z_{q^m-1}: \wt_q(i)=q \}.
\end{align*}
Then
\begin{eqnarray}\label{eqn-apr17}
|D_{(q, m, q-1, \ell)}| & =|\{(i_0, i_1,\ldots, i_{m-1}) \in \{0,1,\cdots,q-1\}^m: i_0+i_1+\cdots+i_{m-1}=1\}|+ \nonumber \\ 
&|\{(i_0,i_1, \ldots, i_{m-1}) \in \{0,1,\cdots, q-1\}^m: i_0+i_1+\cdots+i_{m-1}=q\}|.  
\end{eqnarray} 
Clearly, we have 
\begin{eqnarray}\label{eqn-apr18}
|\{(i_0,i_1, \ldots, i_{m-1}) \in \{0,1,\cdots, q-1\}^m: i_0+i_1+\cdots+i_{m-1}=1\}|=m. 
\end{eqnarray}
It follows Lemma \ref{Sun-1} that
\begin{align*}
&|\{(i_0,i_1, \ldots, i_{m-1}) \in \{0,1,\cdots, q-1\}^m: i_0+i_1+\cdots+i_{m-1}=q\}|\\
&=N(q, m, q-1)=\binom{m+q-1}{q}-m.
\end{align*} 
It then follows from (\ref{eqn-apr17}) and (\ref{eqn-apr18}) that 
\begin{eqnarray*} 
|D_{(q, m, q-1, \ell)}|=\binom{m+q-1}{q}.
\end{eqnarray*} 
The desired conclusion on the dimension of $\C(q, m, q-1, \ell)$ then follows. By Corollary \ref{thm:38}, 
$$
d(\C(q, m, q-1, \ell)) =3\cdot q. 
$$ 
By Theorem \ref{thm:45}, $d(\C(q, m, q-1, \ell))^{\bot} = q^{m-2}$. This completes the proof. 
\end{proof}

\begin{theorem}\label{thm-april221}
Let $q \geq 3$ and $m \geq 2$. Then the constacyclic code $\C(q, m, q-1, q-2)$ has parameters 
$$
\left[ \frac{q^m-1}{q-1},  \binom{m+q-3}{q-2},  3 \cdot q^{m-2} \right], 
$$
and the dual code $\C(q,m, q-1, q-2)^\perp$ has parameters 
$$
\left[ \frac{q^m-1}{q-1},   \frac{q^m-1}{q-1} - \binom{m+q-3}{q-2},   q \right].  
$$  
\end{theorem} 

\begin{proof}
Let $\ell=q-2$, then 
\begin{align*}
D_{(q, m, q-1, \ell)}=&\{i\in \Z_N: \wt_q(i)<(q-1)m-(q-1)+1, \wt_q(i)\equiv 1\pmod{q-1} \}.
\end{align*}
It is easy to check that
\begin{align*}
&\{i\in \Z_N: \wt_q(i)=(q-1)m-(q-1)+1 \}\\
=&\left\{N-(i_0+i_1q+\cdots+i_{m-1}q^{m-1}): \sum_{k=0}^{m-1}i_k=q-2 \right\}.
\end{align*}
 It follows that  
\begin{align*}
|\{i\in \Z_N: \wt_q(i)=(q-1)m-(q-1)+1 \}|=\binom{ q-2+m-1}{q-2}=\binom{m+q-3}{q-2}.
\end{align*}
Then 
$$\dim(\C(q, m, q-1, q-2))=\binom{m+q-3}{q-2},$$
and $$\dim(\C(q, m, q-1, q-2)^{\bot})=\frac{q^m-1}{q-1}-\binom{m+q-3}{q-2}.$$
The minimum distance of the code $\C(q, m, q-1, q-2)$ (resp. $\C(q, m, q-1, q-2)^{\bot}$) follows from Corollary \ref{thm:38} (resp. Theorem \ref{thm:45}). This completes the proof.
\end{proof}

Notice that the constacyclic code $\C(4, m, 3, 2)$ and the code $\mathrm{PRM}(4,m, 2)$ are scalar-equivalent, they have the same weight distribution. The weight distribution of $\C(4, m, 3, 2)$ is the given in Theorem \ref{thm-Lishuxing}. The following four examples show that $\C(4, m, 3, 2)$ is a $(m+1)$-weight code for even $m$ and $m$-weight code for odd $m$. This is consistent with the weight distribution of $\C(4, m, 3, 2)$, which is the same as the weight distribution of $\mathrm{PRM}(4, m, 2)$ documented in Theorem \ref{thm-Lishuxing}.
  
\begin{example} 
Let $(q, m, q-1, \ell)=(4, 2, 3, 2)$. Let $\beta$ be the primitive element of $\gf(4^2)$ with $\beta^4+\beta+1=0$. 
Then $\C(4,2,3,2)$ has parameters $[5,3,3]$ and weight enumerator 
$
1+ 30z^3+15z^4+18z^5. 
$  
Furthermore, $\C(4,2,3,2)^\perp$ has parameters $[5,2,4]$. Both codes are MDS and optimal. 
\end{example} 

\begin{example} 
Let $(q, m, q-1, \ell)=(4, 3, 3, 2)$. Let $\beta$ be the primitive element of $\gf(4^3)$ with $\beta^6 + \beta^4 + \beta^3 + \beta + 1=0$. 
Then $\C(4,3,3,2)$ has parameters $[21, 6, 12]$ and weight enumerator 
$$
1+ 630z^{12} + 3087z^{16} + 378z^{20}. 
$$ 
Notice that $\C(4,3,3,2)$ is distance-optimal \cite{Grassl}. 
Furthermore, $\C(4,3,3,2)^\perp$ has parameters $[21, 15, 4]$ and is almost-distance-optimal \cite{Grassl}.  
\end{example} 

\begin{example} 
Let $(q, m, q-1, \ell)=(4, 4, 3, 2)$. Let $\beta$ be the primitive element of $\gf(4^4)$ with $\beta^8 + \beta^4 + \beta^3 + \beta^2 + 1=0$. 
Then $\C(4,4,3,2)$ has parameters $[85, 10, 48]$ and weight enumerator 
$$
1+ 10710z^{48} + 411264z^{60} +  257295z^{64} +  362880z^{68} + 6426z^{80}. 
$$ 
Furthermore, $\C(4,4,3,2)^\perp$ has the best known parameters $[85,75,4]$ \cite{Grassl}. 
\end{example} 

\begin{example} 
Let $(q, m, q-1, \ell)=(4,5, 3, 2)$. Let $\beta$ be the primitive element of $\gf(4^5)$ with $\beta^{10} + \beta^6 + \beta^5 + \beta^3 + \beta^2 + \beta + 1=0$. 
Then $\C(4,5, 3, 2)$ has parameters $[341, 15, 192]$ and weight enumerator 
$$
1+ 173910z^{192} + 140241024z^{240}+ 809480463z^{256}+ 123742080z^{272} +  104346z^{320}. 
$$ 
Furthermore, $\C(4,5,3,2)^\perp$ has parameters $[341, 326, 4]$. 
\end{example}

\begin{theorem}
Let $q\geq 5$ be an odd prime power, $m\geq 2$, and $r=2$. Let $\ell=(q-1)(m-2)+\ell_0$,	 where $1\leq \ell_0\leq q-2$ and $\ell_0$ is odd. Then the constacyclic code $\C(q, m, 2, \ell) $ has parameters 
$$ 
\left[  \frac{q^m-1}{2},  \frac{q^m-1}{2}-\kappa, \left(\frac{q-\ell_0}2\right) q   \right], 
$$
where 
$$\kappa=\begin{cases}
\sum_{t=0}^{\frac{2q-5-\ell_0}2}\binom{2t+m}{2t+1}-m\sum_{t=\frac{q-1}2}^{\frac{2q-5-\ell_0}2}\binom{2t-q+m}{2t-q+1}~&{\rm if}~\ell_0<q-2,\\
\sum_{t=0}^{\frac{q-3}2} \binom{2t+m}{2t+1}~&{\rm if}~\ell_0=q-2.	
\end{cases}
 $$
Moreover, the dual code $\C(q, m, 2, \ell )^{\bot}$ has parameters 
$$ 
\left[  \frac{q^m-1}{2},  \kappa, d   \right], 
$$
where 
$$d=\begin{cases}
\left(\frac{3+\ell_0}2\right)q^{m-2}~&{\rm if}~\ell_0<q-2,\\
q^{m-1}~&{\rm if}~\ell_0=q-2.	
\end{cases}
 $$
\end{theorem}

\begin{proof}
Note that $(q-1)m-\ell= 2(q-1)-\ell_0$. It is easy to see that 
\begin{align*}
D_{(q, m, 2, \ell)}&=\{i\in \Z_{N}: \wt_q(i)\leq 2q-4-\ell_0,\ \wt_q(i)\equiv 1\pmod{2} \}\\
&= \bigcup_{t=0}^{\frac{2q-5-\ell_0}2}\{i\in \Z_{N}: \wt_q(i)=2t+1 \}.
\end{align*}
If $\ell_0=q-2$, then $2q-4-\ell_0=q-2$. It follows that 
\begin{align*}
|D_{(q, m, 2, \ell)}|&= \sum_{t=0}^{\frac{2q-5-\ell_0}2} \binom{2t+m}{2t+1}.
\end{align*}
If $\ell_0<q-2$, then $2q-4-\ell_0\geq q$. Then
\begin{align*}
|D_{(q, m, 2, \ell)}|&=\sum_{t=0}^{\frac{q-3}2}|\{i\in \Z_{N}: \wt_q(i)=2t+1 \}|+\sum_{t=\frac{q-1}2}^{\frac{2q-5-\ell_0}2}\{i\in \Z_{N}: \wt_q(i)=2t+1 \}\\
&=\sum_{t=0}^{\frac{q-3}2}\binom{2t+m}{2t+1}+\sum_{t=\frac{q-1}2}^{\frac{2q-5-\ell_0}2}\left[\binom{2t+m}{2t+1}-m \binom{2t-q+m}{2t-q+1}\right]\\
&=\sum_{t=0}^{\frac{2q-5-\ell_0}2}\binom{2t+m}{2t+1}-m\sum_{t=\frac{q-1}2}^{\frac{2q-5-\ell_0}2}\binom{2t-q+m}{2t-q+1}.
\end{align*}
It follows that $\dim(\C(q, m,2, \ell))=n-|D_{(q, m, 2, \ell)}|=n-\kappa$, and  $\dim(\C(q, m,2, \ell)^{\bot})=\kappa$. The minimum distance of the code $\C(q, m, 2, \ell )$ (resp. $\C(q, m, 2, \ell)^{\bot}$) follows from Corollary \ref{thm:35} (resp. Theorem \ref{thm:44}). This completes the proof.
\end{proof} 

\begin{theorem}
Let $q$ be an odd prime power, $m\geq 2$ be an integer and $q^m\equiv 1\pmod{4}$. Let $r=2$ and $\ell=\frac{(q-1)m}{2}-1$. Then the code $\C(q, m, 2, \ell)$ is a self-dual code with parameters $$\left[\frac{q^m-1}2, \frac{q^m-1}4, d  \right ],$$ 
where 
\begin{align*}
	d=\begin{cases}
	q^{\frac{m}2}~&{\rm if}~m~{\rm is~even},\\
	(\frac{q+3}{4})q^{\frac{m-1}2}~&{\rm if}~m~{\rm is~odd}.
\end{cases}
\end{align*}
\end{theorem}

\begin{proof}
	If $q^m\equiv 1\pmod{4}$, then $\ell=\frac{(q-1)m}2-1$ is odd. It follows from Theorem \ref{thm:40} that 
	$$\C(q, m, 2, \ell)^{\bot}=\C(q, m, 2, \ell).$$ 
	Now consider the minimum distance of the code $\C(q, m, 2, \ell)$. There are two cases.
	\begin{enumerate}
	\item $m$ is even. Then $\ell=(q-1)(\frac{m}2-1)+(q-2)$. By Corollary \ref{thm:35}, $d(\C(q, m, 2, \ell))=q^{\frac{m}2}$.
	\item $m$ is odd. Then $q\equiv 1\pmod{4}$. Consequently, $\ell=(q-1)(\frac{m-1}2)+\frac{q-3}2$. By Corollary \ref{thm:35}, $d(\C(q, m, 2, \ell))=(\frac{q+3}2)q^{\frac{m-1}2}$.	
	\end{enumerate}
	This completes the proof.
\end{proof}

\begin{example}
Let $(q, m, r, \ell)=(5,3,2,5)$. Let $\beta$ be the primitive element of $\gf(5^3)$ with $\beta^3+3\beta+3=0$. Then the code $\C(5,3,2,5)$ has parameters $[62,31,10]$ and is self-dual.
\end{example}

More examples of the code $\C(q,m,\ell)$ and its dual are the following. 

\begin{example} 
Let $(q, m, r, \ell)=(4, 3, 3, 5)$. Let $\beta$ be the primitive element of $\gf(4^3)$ with $\beta^6 + \beta^4 + \beta^3 + \beta + 1=0$. Then the code $\C(4, 3, 3, 5)$ has parameters $[21, 18, 3]$ and is distance-optimal \cite{Grassl}, and $\C(4, 3, 3, 5)^\perp$ has parameters $[21, 3, 16]$ and is distance-optimal \cite{Grassl}.  
\end{example} 

\begin{example} 
Let $(q, m, r, \ell)=(4, 3, 3, 4)$. Let $\beta$ be the primitive element of $\gf(4^3)$ with $\beta^6 + \beta^4 + \beta^3 + \beta + 1=0$. Then the code $\C(4, 3, 3, 4)$ has parameters $[21, 6, 12]$ and is distance-optimal \cite{Grassl}, $\C(4, 3, 3, 4)^\perp$ has parameters $[21, 15, 4]$.  
\end{example} 

\begin{example} 
Let $(q, m, r, \ell)=(5, 3, 4, 3)$. Let $\beta$ be the primitive element of $\gf(5^3)$ with $\beta^3 + 3\beta + 3=0$. Then the code $\C(5, 3, 4, 3)$ has the best-known parameters $[31, 10, 15]$\cite{Grassl} and $\C(5, 3, 4, 3)^\perp$ has parameters $[31,21,5]$.  
\end{example} 

\begin{example} 
Let $(q, m, r, \ell)=(5, 3, 4, 7)$. Let $\beta$ be the primitive element of $\gf(5^3)$ with $\beta^3 + 3\beta + 3=0$. Then the code $\C(5, 3, 4, 7)$ has parameters $[30, 28, 3]$ and is distance-optimal \cite{Grassl}, and  $\C(5, 3, 4, 7)^\perp$ has parameters $[31,3,25]$ and is distance-optimal \cite{Grassl}.  
\end{example} 

These examples above show that the code $\C(q, m, r, \ell)$ could be optimal in some cases. Thus, the code $\C(q, m, r, \ell)$ is interesting in terms of its error-correcting capability. 

\subsection{Some differences between the codes $\C(q, m, r, \ell)$ and the nonprimitive generalized Reed-Muller codes $\textrm{NGRM}(q, m, r, h)$} 

On one hand, by Corollary \ref{thm:35}, the constacyclic code $\C(q, m, 2, \ell)$ has minimum distance $$d_1=\left(\frac{q-\ell_0}2\right)q^{m-1-\ell_1},$$ where $\ell=(q-1)\ell_1+\ell_0$, $\ell_1\geq 0$, $0\leq \ell_0\leq q-2$ and $\ell_0\equiv 1\pmod{2}$. According to Theorem \ref{eqn-dimPGRMc}, the code $\textrm{NGRM}(q, m, r, h)$ has minimum distance $d_2=\frac{(q-h_0)q^{m-1-h}-1}{2}$, where $(q-1)h+h_0<(q-1)m$, $0\leq h_0\leq q-2$ and $h_0\equiv 0\pmod{2}$. It is easily checked that $d_1=d_2$ if and only if $\ell_1=h=m-1$ and $\ell_0=h_0+1$. Consequently, the constacyclic code $\C(q, m, 2, \ell)$ and the nonprimitive 
generalized Reed-Muller codes $\textrm{NGRM}(q, m, 2, h)$ are different in general.

On the other hand, by Corollary \ref{thm:38}, the constacyclic code $\C(q, m, q-1, (q-1)\ell_1+q-2)$ has minimum distance $d_3=3\cdot q^{m-2-\ell_1}$, where $0\leq \ell_1\leq m-2$. According to Theorem \ref{eqn-dimPGRMc}, the code $\textrm{NGRM}(q, m, q-1, h)$ has minimum distance $d_4=\frac{q^{m-h}-1}{q-1}$, where $0\leq h\leq m-1$. Note that $q\geq 3$, we have $d_3\neq d_4$. Consequently, the constacyclic code $\C(q, m, q-1, (q-1)\ell_1+q-2)$ and the nonprimitive generalized Reed-Muller codes $\textrm{NGRM}(q, m, q-1, h)$ are different.

\section{Summary and concluding remarks}\label{sec4}
   	 
The main contributions of this paper are the constructions and analyses of the two classes of constacyclic codes 
$\C'(q, m, r, \ell)$ and $\C(q, m, r, \ell)$ of length $n=(q^m-1)/r$. These codes are quite interesting in theory as they contain optimal codes and codes with best known parameters (see the examples presented in this paper).  
A new infinite family of distance-optimal constacyclic codes was obtained (see Corollary \ref{cor-2}).  
Since the codes presented in this paper are constacyclic, efficient 
decoding algorithms of them may be obtained by modifying some efficient decoding algorithms of cyclic codes. 
Hence, the codes presented in this paper would also be interesting in practice. 

The automorphism group of the code $\PRM(q, m, \ell)$ was settled in \cite{Berger02}. Hence, the automorphism group of 
the code $\C(q, m, r, \ell)$ is known when $r=q-1$ but is open if $r< q-1$. Note that $\C(q, 2, q-1, q-2)$ is MDS. It follows from Lemma 6 in \cite{Berger02} that the codewords of each fixed nonnzero weight in 
$\C(q, m, q-1, (q-1)\ell_1+q-2)$ support a $2$-design for $m \geq 3$ and $0 \leq \ell_1 \leq m-2$ \cite{dingtang2022}. 
Hence, the codes $\C(q, m, q-1, (q-1)\ell_1+q-2)$ are also interesting from the viewpoint of combinatorics. 

The two classes of constacyclic codes $\C'(q,m,r,\ell)$ and $\C(q,m,r,\ell)$ treated in this paper are new in general.  The constacyclic 
code $\C(q,m,q-1,\ell_1(q-1)+q-2)$ is scalar-equivalent to the linear code $\PRM(q,m, \ell_1(q-1)+q-2)$ and gives 
a constacyclic construction of the code  $\PRM(q,m, \ell_1(q-1)+q-2)$. Hence, another contribution of this paper is the 
proof of the fact that the subclass of projective Reed-Muller codes $\PRM(q,m, \ell_1(q-1)+q-2)$ are constacyclic 
up to equivalence.  
However, it is open if the code 
$\PRM(q,m, \ell)$ is monomially-equivalent to a constacyclic code when $\ell$ is not of the form 
$\ell =\ell_1(q-1)+q-2$. 

It would be very interesting to settle the open problems presented in this paper and determine the automorphism groups of the two classes of constacyclic codes.

\section{Acknowledgements} 
All the code examples in this paper were computed with the Magma software package.

\end{document}